\documentclass[journal,twoside,web]{ieeecolor}
\usepackage{generic}
\usepackage{cite}
\usepackage{amsmath,amssymb,amsfonts}
\allowdisplaybreaks
\usepackage{graphicx}
\usepackage{textcomp}
\usepackage{mathtools}
\usepackage{dsfont}
\usepackage{flushend}
\usepackage{enumerate}
\usepackage{dsfont}
\usepackage{pifont}
\usepackage{algpseudocode}
\usepackage{array}
\usepackage{booktabs}
\usepackage{textcomp}
\usepackage{stfloats}
\usepackage{url}
\usepackage{verbatim}
\usepackage{graphicx}
\usepackage[hidelinks]{hyperref}
\usepackage{subcaption}
\usepackage[affil-it]{authblk}
\usepackage{fancyhdr}
\usepackage{bm}
\usepackage[linesnumbered,ruled,vlined]{algorithm2e}
\SetKwInput{KwInput}{Input}
\SetKwInput{KwOutput}{Output}
\SetKw{KwAssert}{Assert}
\usepackage{centernot}

\usepackage{amsthm}
\newtheorem{lemma}{Lemma}[]
\newtheorem{theorem}{Theorem}[]
\newtheorem{condition}{Condition}[]
\newtheorem{problem}{Problem}[]
\newtheorem{proposition}{Proposition}[]

\theoremstyle{definition}
\newtheorem{definition}{Definition}[]
\newtheorem{oldremark}{Remark}[]
\newtheorem{oldexample}{Example}[]

\newtheorem{assumption}{Assumption}[]

\newenvironment{remark}
{\begin{oldremark}\renewcommand{\qedsymbol}{$\triangle$}\pushQED{\qed}}
	{\popQED\end{oldremark}}
\newenvironment{example}
{\begin{oldexample}\renewcommand{\qedsymbol}{$\triangle$}\pushQED{\qed}}
	{\popQED\end{oldexample}}

\usepackage{tikz}
\usetikzlibrary{shapes,arrows}
\usetikzlibrary{arrows,calc,positioning}
\tikzset{
	block/.style = {draw, rectangle,
		minimum height=1.2cm,
		minimum width=1.2cm},
	input/.style = {coordinate,node distance=1cm},
	output/.style = {coordinate,node distance=2cm},
	arrow/.style={draw, -latex,node distance=2cm},
	pinstyle/.style = {pin edge={latex-, black,node distance=1cm}},
	sum/.style = {draw, circle, node distance=1cm},
}
\definecolor{backgreen}{HTML}{E9F3DF}
\definecolor{backblue}{HTML}{DAE5EC}
\definecolor{backpurple}{HTML}{E5E0E8}
\definecolor{backorange}{HTML}{F2E9DF}
\definecolor{lightergray}{gray}{0.9}
\definecolor{lightgray}{gray}{0.85}

\usepackage[framemethod=TikZ]{mdframed}
\mdfdefinestyle{callout}{%
	linecolor=black,
	linewidth=1pt,%
	roundcorner=0pt,
	innertopmargin=4pt,
	innerbottommargin=4pt,
	innerrightmargin=5pt,
	innerleftmargin=5pt,
	leftmargin = 0pt,
	rightmargin = 0pt,
	backgroundcolor=lightergray
}

\def\BibTeX{{\rm B\kern-.05em{\sc i\kern-.025em b}\kern-.08em
		T\kern-.1667em\lower.7ex\hbox{E}\kern-.125emX}}
\usepackage{dsfont}

\usepackage{dirtytalk}

\usepackage{endnotes}

\definecolor{cadetblue}{rgb}{0.33, 0.41, 0.58}

\newcommand{\dint}[0]{\mathrm{d}}
\DeclareMathOperator{\argmin}{\mathrm{argmin}}
\DeclareMathOperator{\argmax}{\mathrm{argmax}}
\DeclareMathOperator{\EV}{\mathds{E}}

\DeclareMathOperator{\supp}{\mathrm{supp}}

\DeclareMathOperator{\col}{\mathrm{col}}

\DeclareMathOperator{\sign}{\mathrm{sgn}}

\DeclareMathAlphabet{\doublestruck}{U}{BOONDOX-ds}{m}{n}
\newcommand{\ones}[0]{\mathds{1}}
\newcommand{\zeros}[0]{\doublestruck{0}}

\newcommand{\NE}[0]{\mathrm{NE}}

\newcommand{\N}[0]{\mathbb{N}}

\newcommand{\Z}[0]{\mathbb{Z}}
\newcommand{\R}[0]{\mathbb{R}}

\newcommand{\Rnn}[0]{\mathbb{R}_{\geq0}}

\newcommand{\Acal}[0]{\mathcal{A}}
\newcommand{\Bcal}[0]{\mathcal{B}}
\newcommand{\Ccal}[0]{\mathcal{C}}

\newcommand{\Fcal}[0]{\mathcal{F}}
\newcommand{\Gcal}[0]{\mathcal{G}}

\newcommand{\Kcal}[0]{\mathcal{K}}

\newcommand{\Ocal}[0]{\mathcal{O}}
\newcommand{\Pcal}[0]{\mathcal{P}}

\newcommand{\Rcal}[0]{\mathcal{R}}
\newcommand{\Scal}[0]{\mathcal{S}}

\newcommand{\Ucal}[0]{\mathcal{U}}

\newcommand{\Xcal}[0]{\mathcal{X}}

\usepackage{xstring}
\newcommand{\weblink}[1]{%
	\StrSubstitute{#1}{https://}{}[\strippedurl]%
	\href{#1}{{\small \color[RGB]{0,0,0} \url{\strippedurl}}}%
}

\newcommand{\Rd}[0]{R_\mathrm{d}}
\newcommand{\Rr}[0]{R_\mathrm{r}}
\newcommand{\Rno}[0]{R_\mathrm{n}}

\newcommand{\MSNE}[0]{\mathrm{MSNE}}

\newcommand{\UDcbar}[0]{\bar{\Ucal}^c_D}
\newcommand{\UDbar}[0]{\bar{\Ucal}_D}

\newcommand{\Jeff}[0]{J_\mathrm{Eff}}

\newcommand{\JF}[0]{J_\mathrm{F}}

\newcommand{\Rsf}[0]{\mathsf{R}}

\renewcommand{\baselinestretch}{0.96}

\begin{document}
	\title{Token Economy for Fair and Efficient Dynamic Resource Allocation in Congestion Games}
	\author{Leonardo Pedroso, Andrea Agazzi, W.P.M.H. (Maurice) Heemels, Mauro Salazar	\vspace{-1cm}
		\thanks{This work was supported in part by the Eindhoven Artificial Intelligence Systems Institute (EAISI).}
		\thanks{L.~Pedroso, W.P.M.H.~Heemels and M.~Salazar are with the Control Systems Technology section, Department of Mechanical Engineering, Eindhoven University of Technology, The Netherlands (e-mail: \{l.pedroso,m.heemels,m.r.u.salazar\}@tue.nl).}
		\thanks{A.~Agazzi is with the Department of Mathematics and Statistics, University of Bern, Switzerland (e-mail: andrea.agazzi@unibe.ch).}}
	\maketitle
	\vspace{-0.4cm}
	
	\begin{abstract} 
		Self-interested behavior in sharing economies often leads to inefficient aggregate outcomes compared to a centrally coordinated allocation, ultimately harming users. Yet, centralized coordination removes individual decision power. This issue can be addressed by designing rules that align individual preferences with system-level objectives. Unfortunately, rules based on conventional monetary mechanisms introduce unfairness by discriminating among users based on their wealth. To solve this problem, in this paper, we propose a token-based mechanism for congestion games that achieves efficient and fair dynamic resource allocation. Specifically, we model the token economy as a continuous-time dynamic game with finitely many boundedly rational agents, explicitly capturing their evolutionary policy-revision dynamics. We derive a mean-field approximation of the finite-population game and establish strong approximation guarantees between the mean-field and the finite-population games. This approximation enables the design of integer tolls in closed form that provably steer the aggregate dynamics toward an optimal efficient and fair allocation from any initial condition.
	\end{abstract}

	\begin{IEEEkeywords}
		Mechanism design, Sharing economy, Token economy, Congestion games, Mean field games
	\end{IEEEkeywords}

\vspace{-0.1cm}
\section{Introduction}\label{sec:introduction}

Recent advances in the internet of things and connectivity have fueled the rise of sharing economies. In such systems, users compete for access to shared resources (e.g., mobility infrastructures, cloud computing services, and electrical power). To maximize the societal value of these resources, it is essential to design principled rules for their allocation. This challenge falls under the umbrella of mechanism design~\cite{ChremosMalikopoulos2024}, and has recently been identified as a pressing societal concern for the control community~\cite{AnnaswamyJohanssonEtAl2023}.  In this work, we are particularly interested in a resource allocation problem setting of a population of users that want to repeatedly choose a set of available resources to meet their needs (e.g., choosing a path of roads to drive between two points). A user’s reward resulting from selecting a set of resources depends solely on the number of users using those resources (e.g., the travel time increases with the congestion of the road links). 

In these settings, each user selfishly chooses resources to maximize their own reward. It is well-established that such self-interested decisions generally lead to inefficient aggregate outcomes~\cite{Dubey1986}. Specifically, there exist centralized allocations that would increase the average reward across the population. A classic illustration of this inefficiency is Pigou’s two-arc network~\cite[Chap.~17]{NisanRoughgardenEtAl2007}. More broadly, inefficiencies persist even when societal costs differ from average reward. While a central allocation could eliminate inefficiencies, it relies on an authoritarian paradigm and fails to accommodate users’ heterogeneous and time-varying needs. This naturally raises the following question: Can we design mechanisms that align selfish user behavior with societally optimal outcomes? Over the past 70 years, this question has been extensively studied, with most solutions relying on monetary pricing mechanisms~\cite[Chap.~18]{NisanRoughgardenEtAl2007}.

However, monetary approaches embody a profit-driven logic and are often unsuitable in contexts where equitable access, rather than profit, is the priority (e.g., in mobility services). In such scenarios, monetary tools are inherently unfair because they disadvantage low-income users. This motivates the need to balance efficiency with fairness. In single-shot allocation problems, these two objectives conflict, since at the societal optimum some users inevitably face lower rewards than others. In contrast, repeated allocation settings offer an opportunity to reconcile efficiency and fairness through turn-taking: users can alternate in bearing the lower reward, ensuring a fair distribution over time.

A mechanism that induces such pivotal turn-taking behavior was recently introduced in~\cite{CensiBolognaniEtAl2019}. It is based on tokens that are non-tradable, decoupled from real money, and that users spend and earn by accessing the resources. Each user has a wallet of tokens and may only choose sets of resources they can afford. Building on this principle, several studies have explored token-based incentive schemes, mainly focusing on auction-based mechanisms \cite{JalotaPavoneEtAl2020,ElokdaCenedeseEtAl2022,ElokdaBolognaniEtAl2023}, which require users to submit bids every time they request access to a set of resources. This can induce decision fatigue and leaves users uncertain about whether they will obtain the desired resource or be outbid. Our approach deviates significantly from such mechanisms: we propose simple token payment transactions where each set of resources has a fixed token toll. This eliminates the need for bidding and grants users full decision freedom, provided they have a sufficient amount of tokens to pay for the chosen set of resources. Preliminary work \cite{SalazarPaccagnanEtAl2021,PedrosoHeemelsEtAl2023KarmaParallel,PedrosoAgazziEtAl2024EqtEql} has focused on the problem of designing token tolls in simple settings. In \cite{SalazarPaccagnanEtAl2021} and \cite{PedrosoAgazziEtAl2024EqtEql} the two-resource setting is addressed and a closed-form expression for the optimal tolls is derived. In \cite{PedrosoHeemelsEtAl2023KarmaParallel} a generic number of resources is considered, but the user's choices are limited to using one resource at a time, which is not suitable to network congestion games, for instance. Furthermore, the toll design procedure therein becomes intractable as the number of resources increases. In this work, for the first time in the literature, we address a generic resource allocation setting with multiple resources whereby the agents can choose to use a subset of the resources to satisfy their needs. We derive a closed-form expression for optimal tolls and show global convergence to the system optimum.

\subsection{Statement of Contributions}

The main contributions of this paper are threefold:

\begin{itemize}
 \item We model the token economy as dynamic game of a finite population of agents whose decision dynamics are described by a family of evolutionary models \cite{Sandholm2010}. Evolutionary models (initially introduced in \cite{SmithPrice1973,Smith1982}) are a powerful tool to analyze games beyond the classical concept of Nash equilibria by softening assumptions of rationality and knowledge of
 the game and of the equilibrium by agents.
 \item We introduce a mean-field approximation of the token economy model and establish that it is a good approximation of the finite-population model. We show that a notion of equilibrium exists for the mean-field model and that it is essentially unique, i.e., the resource utilization flows of all equilibria are unique. Furthermore, we show that the state of the token economy converges close to the set of equilibria from any initial condition.
 \item We obtain a closed-form solution for optimal integer tolls that drive the state of the token economy close to the system optimum in terms of fairness and efficiency.
\end{itemize}

To summarize, for the first time in the literature, we propose a design procedure for a token economy on a generic resource allocation setting. We show that, under the designed tolls, the system optimum is an equilibrium, and for any initial condition such an equilibrium will be reached.

%

\subsection{Notation}	
For $N\in \N$, the set of consecutive positive integer numbers $\{1,2,\ldots, N\}$ is denoted by $[N]$.
The $i$th entry of a vector $x\in \R^n$ is denoted by $x_i$. The Euclidean norm of a vector $x\in \R^n$ is denoted by $\|x\|$. The $n$-dimensional vector of zeros and ones are denoted by $\zeros_{n}$ and $\ones_n$, respectively. Alternatively, $\zeros$ and $\ones$ denote the vectors of zeros and ones of appropriate dimensions, respectively. The sign of $x\in \R$ is denoted by $\sign(x)$ and takes the values of $-1$, $0$, or $1$ if $x<0$, $x=0$, or $x>0$, respectively. The column-wise concatenation of a finite number of vectors $x^1, x^2, \ldots, x^K$ is denoted by $\col(x^1, x^2, \ldots, x^K)$. The indicator function of $a\in \Xcal$ is denoted by $\delta_a:\Xcal \to \{0,1\}$ and $\delta_a(x) = 0$ if $x\neq a$ and $\delta_a(x) = 1$ if $x = a$. The support of a function $f:\Xcal\to \R$ is denoted by $\supp(f):=\{x\in \Xcal: f(x)\neq 0\}$. 
The interior of a set $\Acal$ is denoted by $\mathrm{int}(A)$ and its cardinality by $|\Acal|$. Given sets $\Xcal_1,\Xcal_2,  \ldots,  \Xcal_K$,  the Cartesian product $\Xcal_1 \times \Xcal_2 \times \cdots \times \Xcal_K$ is denoted by $\bigtimes_{k= 1}^{K}\Xcal_k$. The expected value of a random variable (r.v.) $Z$ is denoted by $\EV[Z]$. The set of all Borel probability measures on $\Acal$ is denoted by $\Pcal(\Acal)$ and the probability of an event $A$ is denoted by $\mathbb{P}(A)$. Given a probability measure $\eta  \in \Pcal(\Acal)$, the mass on $a\in \Acal$ is denoted by $\eta(a)$.
In this paper, to characterize the distribution of total mass of a population of mass $m>0$ over elements of a finite set $\Acal$ we use vectors $\mu \in X_{\Acal}:= \{ \nu \in  \Rnn^{|\Acal|} : \ones^\top \nu = m\}$. For the sake of clarity, by abuse of notation, the mass on $a\in \Acal$ is denoted by $\mu[a]$ and the mass on a subset $\Bcal\subseteq \Acal$ is denoted by $\mu[\Bcal]:= \sum_{a\in \Bcal}\mu[a]$.



\section{Token Economy Model}\label{sec:model}

In this section, we introduce the model of a token economy as a continuous-time dynamic game of (finitely) many agents, where each agent has limited rationality and knowledge. 


\subsection{Finite-Population Dynamic Congestion Game Model}\label{sec:finite_model}

The token economy is modeled as a continuous-time dynamic congestion game described by:
\begin{itemize}
	\item \emph{Population}: Consider a population of $N$ agents. There are $C$ classes of agents with similar needs. We denote the class of a player $i\in [N]$ by $c^i \in [C]$. The set of players that are in each class $c\in [C]$ is denoted by $\Ccal_c :=\{i\in [N]: c^i= c\}$.  The mass of players in class $c\in \Ccal$ is denoted by $m^c := |\Ccal_c|/N$.
	
	\item \emph{Time}: Each player uses the resources in continuous time. Each player $i\in \Ccal_c$ is equipped with a Poisson clock with rate $\Rd>0$. Each time the clock of a player rings, they want to use shared resources to meet their needs. We denote the time of the $k$-th clock ring of a player $i\in [N]$ by a r.v.\ $t^i_k$.
	
	\item \emph{Resources}: There are $\!\Rsf\!$ resources. When the clock of an agent rings they can choose a set of the resources, which is called an action, to satisfy their needs. The set of actions that satisfy the needs of agent $i\!\in\! \Ccal_c$ is denoted by $\Acal^c \!\subseteq \!2^{[\Rsf]}$, where $ 2^{[\Rsf]}$ is the power set of $[\Rsf]$ excluding the empty set, i.e., the set of all combinations of elements of $[\Rsf]$. The set of all actions is denoted by $\Acal\! := \cup_{c\in [C]} \Acal^c$. The set of actions available to agents in class $c\!\in \![C]$ that use a resource $r\!\in\! [\Rsf]$ is denoted by $\Acal^c_{[\Rsf]}(r)\!\subseteq \!\Acal^c$.
	
	\item \emph{States}: At each time $t$, each player $i\in [N]$ is characterized by a state that is the amount of tokens they possess. It is characterized by a r.v.\ $k^i(t)$, which is an integer in $\Kcal := \{0,1,\ldots,\bar{k}\}$, where $\bar{k} \in \N$ is a maximum amount of tokens allowed.
	
	\item \emph{Tolls}:  A generic tolling mechanism for each class $c\in [C]$ is a map $\tau^c:\Acal^c \to \Z$ that associates each action $a\in \Acal^c$ with an integer token toll $\tau^c(a)$. The collection of toll maps of all classes is denoted by $\tau:= (\tau^c)_{c\in [C]}$.
	
	
	\item \emph{Actions}: When the clock of player $i \in \Ccal_c$ rings at time $t$, they choose an action $a \in \Acal^c$ that they can afford, i.e., $a\in \Acal^c(k^i(t)) := \{a\in \Acal^c : \tau^c(a) \leq k^i(t)\}$. The action that each agent $i\in [N]$ would take at time $t$ if their clock were to ring is characterized by a r.v.\ $a^i(t)$. We also define $\Acal := \cup_{c\in[C]}\Acal^c$, $q^c := |\Acal^c|$, and $q:= \sum_{c\in [C]}q^c$.

	\item \emph{State transitions}: When the clock of a player $i \in \Ccal_c$ rings and they choose an action $a$, their state is updated according to a Markov transition kernel. Specifically, their amount of tokens evolves according to a deterministic Markov kernel $\phi_\Kcal^c: \Kcal \times  \Acal^c \to  \Pcal(\Kcal)$ defined as
	\begin{equation}\label{eq:kernel_K}
		\phi^c_\Kcal(k,a) = 	\begin{cases}
		\delta_{k-\tau^c(a)}(\cdot), & k-\tau^c(a) < \bar{k}\\
		\delta_{\bar{k}}(\cdot), & k-\tau^c(a) \geq \bar{k}.
		\end{cases}
	\end{equation}

	\item \emph{State-action distribution}: The empirical joint state-action distribution of class $c\in [C]$ at time $t$ is characterized by a measure-valued r.v.\ $\hat{\mu}^c_{\Kcal \times \Acal}(t)$ with support in $X^c_{\Kcal \times \Acal}:=\{\nu \in \Rnn^{(\bar{k}+1)q^c} : \ones^\top\nu = m^c\}$. Recall that, by abuse of notation, $\hat{\mu}^c_{\Kcal \times \Acal}[k,a](t)$ is the r.v.\ associated with the mass on $k\in \Kcal$ and $a\in \Acal^c$ and it is defined as the empirical measure $\hat{\mu}^c_{\Kcal \times \Acal}[s,a](t) := \frac{1}{N}\sum_{i\in \Ccal_c} \delta_{s^{i}(t)}(s) \delta_{a^{i}(t)}(a)$. The concatenation of the empirical joint state-action distributions for all classes is denoted by $\hat{\mu}_{\Kcal \times \Acal} = \col(\hat{\mu}^c_{\Kcal \times \Acal}, c\in [C])$ with support in $X_{\Kcal \times \Acal} := \bigtimes_{c\in [C]}X^c_{\Kcal \times \Acal}$. Notice also the the empirical normalized rate at which agents choose each resource $r\in [\Rsf]$ is given at time $t$ by
	\begin{equation}\label{eq:hat_sigma}
		\hat{\sigma}_r(t) = \Rd \sum_{c\in [C]} \sum_{k\in \Kcal}\sum_{a\in \Acal_{[\Rsf]}^c(r)}\hat{\mu}^c_{\Kcal\times \Acal}[k,a](t).
	\end{equation} 
	We also define $\hat{\sigma}(t) := \col(\hat{\sigma}_r(t), r \in [\Rsf])$.

	\item \emph{Single-stage reward}: Each resource $r\in [\Rsf]$ is associated with a function $w_r:\R \to \R$ that characterizes the reward of using resource $r$ based on the rate of agents using it. 
	Moreover, the single-stage reward of taking an action is modeled by a real-valued function $w:  \Acal \times \R^{\Rsf}  \to \R$ and corresponds to the summation of the rewards of the resources that the action uses. Specifically, the single-stage reward of an agent that takes action $a\in \Acal$ at time $t$ is
	\begin{equation*}
		w(a,\hat{\sigma}(t)) = \sum\nolimits_{r\in a} w_r(	\hat{\sigma}_r(t)).
	\end{equation*}
	
	\item \emph{Payoff}: The payoff of an agent $i\in \Ccal_c$ is modeled as  the long-time average reward, which is given by
	\begin{equation}\label{eq:def_Ji}
		J^i := \lim_{T\to \infty}\frac{1}{T}\EV\left[\sum_{k=1}^T w(a^i(t^i_k),\hat{\sigma}(t^i_k))\right].
	\end{equation}%
 \end{itemize}

%
%

\begin{definition}\label{def:cong_network}
	The congestion game is said to be a \emph{network congestion game} if there is a directed graph $\Gcal$ whose edges are the resources $[\Rsf]$, each class is associated with an origin and a destination node, and the actions available to satisfy the needs of each class are the set of all paths that connect the origin to the destination nodes.
\end{definition}



\vspace{-0.4cm}

\subsection{Information Structure}

The information structure of the token economy model defines the information available to each agent during the game. Then, given the information available to an agent they choose one action. We call the map from known information by an agent to their choice of action a \emph{policy}. We model the information structure of the token economy game as:
\begin{itemize}
	\item Each agent \emph{does not have access to aggregate information} of the distributions of the amount of tokens of other agents. Indeed, their amount of tokens is tracked by the operator of the resources but not disclosed publicly.
	\item Each agent does not keep an history of their past token amounts or their past actions. At each clock ring an agent knows only their current amount of tokens. 
	\item Each agent cannot predict how payoffs will evolve in the future. When they choose a map from information to a decision, they plan to use it forever.
\end{itemize}
Policies that are consistent with the above information structure are said to be \emph{oblivious}, \emph{Markov}, and \emph{stationary}, respectively. Such policies are formally characterized, for each class $c\in[C]$, by a map from $\Kcal$ to $\Pcal(\Acal^c)$, i.e., from the amount of tokens of the agent when the clock rings to a randomization of actions. The set of policies that are consistent with the above information structure and the action constraints of an agent in class $c\in [C]$ is given by
\begin{equation*}
	\Ucal^c := \left \{ u:\Kcal \to \Pcal(\Acal^c)  \;| \; \supp(u(k)) \subseteq \Acal^c(k) \; \forall k\in \Kcal \right\}.
\end{equation*} 
In general, a policy in $\Ucal^c$ is randomized in the sense that it maps the individual state of the agent to a randomization of available actions. A degenerate case are deterministic policies, which map the states to a single action with probability one. The set of deterministic policies available to an agent in class $c\in [C]$ is defined as
\begin{equation*}
	\Ucal^c_D := \left \{ u \in \Ucal^c \; |\; \forall k\in \Kcal\; \exists a\in \Acal^c(k): \supp(u(k)) = \{a\}\right\}.
\end{equation*}
By abuse of notation, whenever clear from the context, we also write a deterministic probability measure  $u(k) = \delta_a(\cdot)$ of a deterministic policy $u\in \Ucal^c_D$ as $u(s) = a$.


There is considerable debate on whether randomized policies are meaningful in real-life applications (see \cite[Chap.~3.2]{Rosenthal2006} for an insightful discussion). To characterize the agents decisions in a token economy, we consider that at each time $t$ each agent $i\in \Ccal_c$ uses a deterministic policy in $\Ucal_D^c$ that is characterized by the r.v.\ $u^i(t)$. This means that an agent using a policy $u\in \Ucal_D^c$ will choose the same set of resources to satisfy their needs in any circumstances where their amount of tokens is the same. In Section~\ref{sec:ev_dynamics}, we describe a model on how the policy chosen by each agent evolves in time. 

\vspace{-0.4cm}

\subsection{Regularity Assumptions}\label{sec:model_assumptions}

In what follows, we introduce mild regularity conditions on the model that are key for the prescriptive and descriptive analysis of Sections~\ref{sec:descriptive} and~\ref{sec:design}. 

\begin{assumption}\label{ass:affordable_prices}
	For all $c\in [C]$, $\min_{a\in \Acal^c}\tau^c(a) \leq 0$.
\end{assumption}

\begin{assumption}\label{ass:cont}
	For all $r\in [\Rsf]$ the reward function $w_r: \R \to \R$ is Lipschitz continuous and uniformly decreasing.
\end{assumption}

These assumptions impose regularity conditions on the tolls and on the reward functions of the resources. The former is necessary to ensure that agents have always access to affordable actions. The latter is necessary for the existence and uniqueness of equilibria and of solutions to the mean-field game evolution that is modeled in Section~\ref{sec:descriptive}.

Second, we argue that some deterministic policies in $\Ucal_D$ are senseless and, as a result, an agent would not even consider using them. Specifically, \emph{it would be senseless for an agent to choose cheaper actions as their amount of tokens increases}, i.e., it is senseless for agent $i\in \Ccal_c$ to choose $u(k_1) = a_1$ with toll $\tau(a_1)$ when they have $k_1$ tokens and then to choose $u(k_2) = a_2$ with toll $\tau(a_2) < \tau(a_1)$ when they have $k_2>k_1$ tokens. Indeed, on the one hand, if choosing $a_1$ has higher expected long-time average reward than choosing $a_2$ with $k_1$ tokens and $\tau(a_2) < \tau(a_1)$ that is also the case with $k_2>k_1$ tokens, so the agent would use $a_1$ again when they have $k_2$ tokens. On the other hand, if choosing $a_2$ has higher reward than $a_1$ with $k_2$ tokens and $\tau(a_2) <\tau(a_1)$, then the agent would also use $a_2$ when they have $k_1$ tokens. As a result, henceforth, we discard such policies as formally presented in the following assumption.

\begin{assumption}\label{ass:senseless_policies}
	For all $c\in [C]$ an agents $i\in \Ccal_c$ can only choose policies in $\bar{\Ucal}_{D}^c$, where
	\begin{equation*}
		\bar{\Ucal}_{D}^c \!:=\! \left\{ u\!\in\! \Ucal_D^c\,|\, \forall k,k^\prime \!\in\! \Kcal : k^\prime \!>\! k\! \implies\! \tau(u(k^\prime)) \!\geq \!\tau(u(k))  \right\}\!.\!
	\end{equation*}
\end{assumption}


We also define $n^c := |\bar{\Ucal}_{D}^c|$ and $n:= \sum_{c\in [C]}n^c$. Under Assumption~\ref{ass:senseless_policies}, the empirical joint state-policy distribution of class $c\in [C]$ is characterized by a r.v.\ $\hat{\mu}^c(t)$ whose support in $X^c :=\{\nu \in \Rnn^{(\bar{k}+1)n^c} : \ones^\top\nu = m^c\}$, which, by abuse of notation, is given by $\hat{\mu}^c[k,u](t) := \frac{1}{N}\sum_{i\in \Ccal_c}  \delta_{k^{i}(t)}(s) \delta_{u^{i}(t)}(u)$ for all $k\in \Kcal$ and all $u\in \UDcbar$.  The concatenation of the empirical joint state-policy distributions for all classes is denoted by $\hat{\mu} = \col(\hat{\mu}^c, c\in [C])$ with support in $X:= \bigtimes_{c\in [C]}X^c$. Define map $\sigma: X \to \Rnn^\Rsf$ that expresses that resource utilization rates as a function of state-policy distributions $\mu \in X$ as
\begin{equation*}
	\sigma_r(\mu) =  \Rd\sum_{c\in [C]}\sum_{k\in \Kcal}\sum_{a\in \Acal_{[\Rsf]}^c(r)}\sum_{u\in \UDcbar} \mu^c[k,u]u(a|k).
\end{equation*}
for all $r\in [\Rsf]$. Interestingly, Assumption~\ref{ass:senseless_policies} will not take an explicit role in the analysis of Sections~\ref{sec:descriptive} and~\ref{sec:design}. Instead, its main purpose is to support the validity of Assumption~\ref{ass:noise}, which is introduced in what follows. This aspect is analyzed further in Example~\ref{eg:effect_noise}.

Third, a key property to significantly simplify the analysis of the game is that the long-time average reward of an agent does not depend on their initial amount of tokens. This can be achieved under another assumption that requires that the Markov chains of the transitions of the amount of tokens have a unique recurrent communicating class. However, that is not generally possible under the deterministic Markov kernel in \eqref{eq:kernel_K}. To overcome that, we assume that such transitions have small stochastic noise. Specifically, we consider that the operator of the resources maintains a noise Poisson clock with rate $\Rno \ll \Rd$ associated with each agent. Each time the noise clock of an agent rings, the agent is gifted one token unit. Formally, the transition kernel of class of regularization noise Markov chain $	\phi_{n}(k) : \Kcal \to \Pcal(\Kcal)$ is characterized for all $k\in \Kcal$ by
\begin{equation}\label{eq:kernel_noise}
	\phi_{n}(k) = 	\begin{cases}
		\delta_{k+1}(\cdot), & k < \bar{k}\\
		\delta_{\bar{k}}(\cdot), & k = \bar{k}.
	\end{cases}
\end{equation}

The addition of noise is a powerful regularization technique often found in the literature \cite{FreidlinWentzell1998,BorkarSowmyaEtAl2025}. In this case in particular, the aforementioned regularization noise model has a physical interpretation and can actually be implemented. For a detailed overview of elementary Markov chain analysis tools used in this paper see \cite{Norris1997}. This assumption is formally presented below.

\begin{assumption}\label{ass:noise}
	Each agent is associated with a noise Poisson clock of rate $\Rno < \Rd/4$. Each time it rings, the amount of tokens of the agents evolves according to the Markov transition kernel $\phi_{n}$ defined in \eqref{eq:kernel_noise}.
\end{assumption}

Notice that for any agent $i\in \Ccal_i$ that uses any policy $u\in \Ucal_D^c$ the continuous-time Markov chain of the amount of tokens is characterized by two jump chains: (i)~one associated with the use of the resources, i.e., $\phi_{\Kcal}^{c,u}(k) = \sum_{k^\prime \in \Kcal}\sum_{a\in \Acal^c(k^\prime)} \phi^c_{\Kcal}(k^\prime|k,a)u(a|k)\delta_{k^\prime}(\cdot)$, which has rate $\Rd$; (ii)~and another associated with the regularization noise, which is characterized by $\phi_n$ in \eqref{eq:kernel_K} and has rate $\Rno$. By the memoryless property of Poisson processes, the continuous-time Markov chain of the amount of tokens is equivalently described by a single jump chain characterized by 
\begin{equation}\label{eq:equiv_jump_chain_noise}
	\phi^{c,u} = \frac{\Rd}{\Rd+\Rno}\phi_{\Kcal}^{c,u}+\frac{\Rno}{\Rd+\Rno} \phi_n,
\end{equation}
with rate $\Rno+\Rd$. The collection of Markov chains of all classes and all policies is denoted by $\phi:=(\phi^{c,u})_{c\in [C], u\in \UDcbar}$. In the following result, we establish two key properties that follow from Assumption~\ref{ass:noise}.	

\begin{lemma}\label{lem:F_def}
	Under Assumption~\ref{ass:noise}, the continuous-time Markov chain of the amount of tokens of an agent $i\in \Ccal_c$ that uses a policy $u\in \Ucal^c_D$, whose jump chain is defined in \eqref{eq:equiv_jump_chain_noise},  admits a unique stationary distribution denoted by $\eta^{c,u} \in \Pcal(\Kcal)$.
	Moreover, for any fixed state-policy distribution $\mu \in X$, under Assumption~\ref{ass:cont}, the long-time average reward of an agent $i\in \Ccal_c$ using policy $u \in \UDcbar$ does not depend on the initial state and can be expressed as
	\begin{equation}\label{eq:Ji_avg}
		J^i(\mu) =  \sum_{k\in \Kcal} \sum_{a\in \Acal^c} \eta^{c,u}(k)u(a|k)w(a,\sigma(\mu)).
	\end{equation}
\end{lemma}
\begin{proof}
	See Appendix~\ref{sec:proof_lem_F_def}.
\end{proof}

Since there is a finite number of policies for each class, define for all $c\in [C]$ a payoff map $F^c: X\to \R^{|\UDcbar|}$ resorting to \eqref{eq:Ji_avg} as 
	\begin{equation}\label{eq:def_F}
	\!\!F^c(\mu) \!= \!\col\!\left(\sum_{k\in \Kcal} \sum_{a\in \Acal^c} \!\eta^{c,u}(k)u(a|k)w(a,\sigma(\mu)),u\!\in\! \UDcbar\!\right)\!\!.\!
	\end{equation}
	For the sake of clarity, by abuse of notation, we denote the component associated with policy $u\in \UDcbar$ by $F^c_u(\mu)$. We also write the concatenation of the payoff maps of all classes as a payoff map $F^c: X\to \R^{|\UDbar|}$ given by $F(\mu) = \col(F^c(\mu), c\in [C])$. Notice that the model can be fully characterized by the pair $(F,\phi)$.

\begin{example}\label{eg:effect_noise}
	\begin{figure}[b]
		\centering
		\includegraphics[width = 0.7\linewidth]{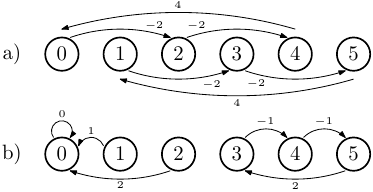}
		\caption{Illustrative token transition chains: Nodes represent the token amounts, edges represent the transition after taking an action and are labeled with the corresponding toll value.}
		\label{fig:illustrative_ass_noise}
	\end{figure}
	In this example, we illustrate why Assumption~\ref{ass:noise} is reasonable and why it is necessary to establish Lemma~\ref{lem:F_def}. There are two distinct reasons why the deterministic token dynamics $\phi^c_\Kcal$ may lead to non-unique steady-state distributions. On the one hand, consider the case where all tolls are even. Then, as illustrated in Fig.~\ref{fig:illustrative_ass_noise}(a), according to the deterministic token transitions of a policy in $\Ucal^c_D$, an agent that starts with an even (odd) amount of tokens will always remain with an even (odd) amount of tokens. As a result, depending on the parity of the initial amount of tokens of an agent, the stationary state distribution will vary. In this case, noise mixes the even and odd chains leading to a unique stationary token amount distribution.
	On the other hand, as illustrated in Fig.~\ref{fig:illustrative_ass_noise}(b), it can also happen that the map between the amount of tokens and actions of a policy in $\Ucal^c_D$ may induce two or more distinct communicating classes that are not intertwined. Notice that, contrarily to the case of Fig.~\ref{fig:illustrative_ass_noise}(a), the policy of Fig.~\ref{fig:illustrative_ass_noise}(b) is \emph{senseless} in the context of Assumption~\ref{ass:senseless_policies}. Moreover, in Fig.~\ref{fig:illustrative_ass_noise}(a) the two chains have identical payoff and their mixing is fast. But that is not the case for Fig.~\ref{fig:illustrative_ass_noise}(b), where the two chains have potentially significantly distinct payoff and will only mix very slowly (since $\Rno \ll \Rd$). For these reasons, in the case of Fig.~\ref{fig:illustrative_ass_noise}(b), even though the regularization noise does lead to a unique recurrent communicating class, the payoff of an agent cannot be accurately be modeled as the long-time average payoff. However, for policies that satisfy Assumption~\ref{ass:senseless_policies}, the long-time average payoff is a reasonable modeling choice.
\end{example}

\vspace{-0.3cm}

\subsection{Evolutionary Decision Model}\label{sec:ev_dynamics}

Thus far we modeled the way an agent chooses an action depending on their amount of tokens making use of the concept of a policy. However, we did not make any considerations on how an agent chooses a policy and how they may switch policies as the aggregate decisions change with time. In this section, we address this question making use of an evolutionary model. Evolutionary models have been thoroughly analyzed for static games, where the rewards do not depend on the evolution of an individual state of each agent (e.g., \cite{Sandholm2010}). Those settings are also known as population games. The token economy model presented in this section is dynamic, therefore the proposed evolutionary model is based on a dynamic version recently proposed in \cite{PedrosoAgazziEtAl2025MFGAvg,PedrosoAgazziEtAl2025MFGAvgII}.

According to the proposed evolutionary model, each agent $i\in \Ccal_c$ uses a deterministic policy $u^i(t) \in \UDcbar$ at time $t$. During the game, agents, who have access to very limited information, will revise their policy choice and possibly switch to another policy. This is expressive of \emph{inertia} and \emph{myopia} properties of behavior seen in real-life. Recall that $\hat{\mu}(t)$ denotes the empirical joint state-policy distribution, and additionally we also denote the empirical marginal policy distribution of class $c\in [C]$ as $\hat{\mu}^c[\Kcal,\cdot](t)$, i.e., $\hat{\mu}^c[\Kcal,\cdot](t) := \col(\hat{\mu}^c[\Kcal,u](t); u\in \UDcbar)$ with $\hat{\mu}^c[\Kcal,u](t) = \sum_{k\in \Kcal}\hat{\mu}^c[k,u](t)$. Formally, the evolutionary model is characterized by
\begin{itemize}
	\item \emph{Time}: Each agent makes policy revisions in continuous-time. Each agent is equipped with a Poisson clock with rate $\Rr \ll \Rd $. Each time the revision clock of an agent rings they revise the policy that they are using. The clocks of different players are independent and have the same rate.  Action, noise, and revision clocks of agents are independent.
	
	\item \emph{Policy transitions}: Upon a revision opportunity of a player, their policy choice evolves according to a \emph{revision protocol}. Let $X^c_{\Ucal_D}:= \{\nu \in \Rnn^{n^c}: \ones^\top\nu = m^c\}$, a revision protocol of a class $c\in [C]$ is a map $\rho^c : \R^{n^c} \times X^c_{\Ucal_D} \to \Rnn^{n^c\times n^c} $, where the component associated with the pair $(u,v) \in \UDcbar \times \UDcbar$ is denoted, by abuse of notation, by $\rho^c_{uv}$. Specifically, a player using policy $u\in \UDcbar$ switches to policy $v\in \UDcbar$ with a switch rate $\rho^c_{uv}(F^c(\hat{\mu}),\hat{\mu}^c[\Kcal,\cdot])$, where $F^c(\hat{\mu})$ is defined in \eqref{eq:def_F} and the policy ordering of $F^c(\hat{\mu})$  and of $\hat{\mu}^c[\Kcal,\cdot]$ is consistent.
\end{itemize}
Intuitively, if agent $i\in \Ccal_c$ using policy $u\in \UDcbar$ receives a revision opportunity, they switch to a policy $v\in \UDcbar$ with probability $\rho^c_{uv}(F^c(\hat{\mu}),\hat{\mu}^c[\Kcal,\cdot])/\Rr$, and they continue to use the same policy with probability $1-\sum_{v\neq u}\rho^c_{uv}(F^c(\hat{\mu}),\hat{\mu}^c[\Kcal,\cdot])/\Rr$. We make an assumption to ensure that the aforementioned switching probabilities are well defined and continuous as follows.

\begin{assumption}\label{ass:rev_protocol}
	For all $c\in [C]$, the revision protocol $\rho^c$ is Lipschitz continuous and for all $u\in \UDcbar$
	\begin{equation*}
		\Rr \geq \max_{\mu \in X} \sum_{v\in \UDcbar \setminus \{u\}} \rho^c_{uv}(F^c(\mu),\mu^c[\Kcal,\cdot]).
	\end{equation*}
\end{assumption}

The revision protocols $\rho^c$ for $c\in [C]$ are characterized analogously to the literature of evolutionary dynamics for static games \cite{Sandholm2010}. Therein, physically meaningful families of Lipschitz revision protocols are defined, which are presented in what follows for the sake of completeness.

\begin{definition}[{Imitative via comparison revision protocol \cite[Chap.~5.4]{Sandholm2010}}]\label{def:imitative}
	Consider a revision protocol $\rho^c$ defined as $\rho^c_{uv}(F,x) = r^c_{uv}(F,x)x_v/m^c$, where  $r^c :\R^{n^c} \times X^c_{\Ucal_D} \to \Rnn^{n^c\times n^c}$ is a Lipschitz continuous conditional imitation rate map with monotone net conditional imitation rates, i.e., $F_v \geq F_u \iff r^c_{kv}(F,x) -   r^c_{vk}(F,x)  \geq  r^c_{ku}(F,x) -   r^c_{uk}(F,x), \forall F \in \R^{n^c}\;\forall x \in X^c_{\Ucal_D} \;\forall u,v,k \in \UDcbar$. Then,  $\rho^c$ is said to be an \emph{imitative} revision protocol.  \hfill$\triangle$
\end{definition}

\begin{definition}[{Pairwise comparison revision protocol \cite[Chap.~5.6]{Sandholm2010}}]\label{def:pairwise_cmp}
	Consider a revision protocol $\rho^c$ defined as $\rho^c_{uv}(F,x) = r^c_{uv}(F)$, where $r^c :\R^{n^c} \to \Rnn^{n^c}$ is a Lipschitz continuous rate map that is sign-preserving, i.e., $\sign(r^c_{uv}(F)) = \sign(\max(0,F_v-F_u)), \forall F
	\in \R^{n^c} \; \forall u,v \in \UDcbar$.
	Then, the revision protocol $\rho^c$ defined as $\rho^c_{uv}(F,x) = r^c_{uv}(F)$ is called a \emph{pairwise comparison} revision protocol. \hfill$\triangle$
\end{definition}

In \cite[Chap.~5.5]{Sandholm2010} a third family of revision protocols is defined. Nevertheless, it requires that each agent has access to the average payoff of the population. Since that is not physically meaningful in a token economy game, we disregard it in this setting.

\begin{remark}
	The families of revision protocols in Definitions~\ref{def:imitative} and~\ref{def:pairwise_cmp} have an intuitive interpretation. First, following an imitative revision protocol, when an agent is given a revision opportunity they choose another agent of the same class at random and imitate their policy with a probability that depends on its payoff. Second, following a pairwise comparison protocol, when an agent is given a revision opportunity they choose a policy available to their class at random and switch with nonzero probability if it strictly increases the agent's current payoff. Notice that information requirements of these families is very mild, since they require knowing, at most, the payoff of the current policy and a randomly chosen policy only.
\end{remark}

\section{Problem Statement}\label{sec:probelm_statement}

With the token economy model presented in Section~\ref{sec:model}, we are ready to state the token economy  design problem in this section.

\begin{mdframed}[style=callout]
	Intuitively, the goal is to design for all $c\in \Ccal$ a toll map $\tau^c: \Acal^c \to \Z$ such that the allocation of resources to the agents is \emph{efficient} and \emph{fair} as time evolves.
\end{mdframed}

In the next two sections, we formalize the concepts of efficiency and fairness in this context. Then, in Section~\ref{sec:design_probelm}, we formalize the design problem.

\subsection{Efficiency}
From the macroscopic system-level perspective, we consider that there is an efficiency metric associated with the aggregate use of the resources that we desire to maximize. Formally, it is characterized by a function $\Jeff: \Rnn^\Rsf\to \R$ of the normalized rate at which agents use the resources, i.e., a function of $\hat{\sigma}$ defined in \eqref{eq:hat_sigma}.

\begin{example}
	The efficiency metric can be for example the average reward of the population, i.e., $\Jeff(\sigma) = \sum_{r\in [\Rsf]}\sigma_r w_r(\sigma_r)$. Another meaningful welfare function is, for example, the negative of the total environmental impact of using the resources.
\end{example}

\subsection{Fairness}

First, the most fundamental fairness consideration in a token economy is that \emph{users with the same needs and the same impact on congestion have the same average reward in time, no matter other external factors such as wealth or social status}. This fairness condition is called intra-class fairness and for class $c\in [C]$ it is formalized resorting to the payoff definition in \eqref{eq:Ji_avg} as
\begin{equation}\label{eq:intra_class_fairness}
	J^i(\hat\mu(t)) = J^j(\hat\mu(t)) \quad \forall i,j \in \Ccal_c.
\end{equation}
Notably, the intra-class fairness condition is in line with recent equity-over-time distributive priciples of justice for mobility systems \cite{QiuReyEtAl2025}. 


Second, as discussed in \cite{PedrosoAgazziEtAl2024EqtEql}, since in a token economy users may have different needs (different action sets $\Acal^c$) and, hence, they contribute differently to the congestion of the resources, one can evaluate inter-class fairness from different perspectives. Moreover, these fairness considerations vary significantly depending on the application of the token economy. For example, in a mobility system many tailored inter-class fairness metrics have been recently proposed \cite{BangDaveEtAl2024,SalazarBetancurEtAl2024Acc,SalazarBetancurEtAl2025}. As a result, the characterization of specific inter-class fairness metrics is out of the scope of this work. Denote the average reward of class $c\in [C]$ by $\bar w_c$, which can be written as a function of a state-policy distribution $\mu \in X$ as 
\begin{equation*}
	\bar w_c(\mu) =  \frac{1}{m^c}\sum_{k\in \Kcal}\sum_{u\in \UDcbar} \sum_{a\in \Acal^c} \mu^c[k,u]u(a|k)w(a,\sigma(\mu)),
\end{equation*}
and its vector concatenation by $\bar w = \col(\bar w_1, \bar w_2,\ldots, \bar w_C)$. Henceforth, we assume the existence of an inter-class fairness metric that is characterized by a function  $\JF: \Rnn^C \to \R$ of $\bar w$, which is in line with the aforementioned literature in a mobility context. It may also be the case that inter-class fairness considerations impose constraints on the average reward of some classes. For simplicity, these constraints are henceforth ignored, although they can be readily incorporated into the proposed design approach if needed. 

\begin{example}
	The inter-class fairness metric that one desires to maximize can be the minimum average payoff across all classes, i.e., $\JF(\bar{w}) = \min_{c\in[C]} \bar w_c$.
\end{example}



\subsection{Design Problem}\label{sec:design_probelm}

Consider any combination of efficiency and inter-class fairness metrics characterized by $g(\Jeff(\sigma),\JF(\bar w))$ for some function $g$, which satisfies the following mild assumption.

\begin{assumption}\label{ass:J}
	The map $\mu \mapsto \,g(\Jeff(\sigma(\mu)),\JF(\bar w(\mu)))$ is continuous w.r.t.\ $\mu$ and bounded in $X$.
\end{assumption} 

We are interested in maximizing the long-time average combination of efficiency and inter-class fairness metrics, i.e., 
\begin{equation}\label{eq:def_J}
	J = \lim_{T\to \infty}\frac{1}{T}\EV\!\left[ \int_0^T \!\!g\big(\Jeff\big(\sigma(\hat{\mu}(t))\big),\JF\big(\bar w(\hat{\mu}(t))\big) \big) \dint t\right]\!\!.
\end{equation}
One can then state the design problem as follows.

\begin{problem}\label{prob:design}
	 Find a time-invariant toll map collection $\tau = (\tau^c)_{c\in [C]}$ that attains, from any initial condition:
	 \begin{enumerate}[(i)]
	 		\item Intra-class fairness i.e., $J^i = J^j$ for all $i,j \in \Ccal_c$ for all $c\in [C]$;
	 		\item Maximum chosen combination of efficiency and inter-class fairness $J$ defined in \eqref{eq:def_J}.
	 \end{enumerate}
\end{problem}

Intuitively, our goal is to design token tolls for using the resources such that: (i)~every agent in the same class perceives the same average payoff; and (ii)~average resource utilization rates of each class converge, from any initial condition, to a feasible maximizer of $g(\Jeff(\sigma(\mu)),\JF(\bar w(\mu)))$.

The approach to this problem is divided into two parts. First, in Section~\ref{sec:descriptive}, we take a \emph{descriptive} approach. Specifically, for any given toll collection $\tau$, we describe the evolution of the decisions of the agents and of the resource utilization rates. Second, in Section~\ref{sec:design}, we take a \emph{prescriptive} approach. Specifically, using the results of the descriptive analysis, we design $\tau$ to prescribe the goals of Problem~\ref{prob:design}.

\section{Descriptive Analysis}\label{sec:descriptive}

In this section, given any toll collection $\tau$ that satisfies Assumption~\ref{ass:affordable_prices}, we describe the evolution of the token economy. Specifically, we use an analogous mean-field model, which has strong approximation properties w.r.t.\ the finite-population model. Unlike the finite-population model, the mean-field model is analytically tractable and one can gather insight into the equilibria of the token economy and their stability. In Section~\ref{sec:mean_field_model} we propose an analogous mean-field model and establish approximation properties w.r.t.\ the finite-population model presented in Section~\ref{sec:finite_model}. In Section~\ref{sec:mf_equilibria} we study the equilibria of the analogous mean-field model and in Section~\ref{sec:mf_stability} we study the evolutionary stability of the equilibria.

\subsection{Mean Field Token Economy Model}\label{sec:mean_field_model}

In this section, we introduce the mean-field approximation of the finite-population model, where the population is infinite, each agent carries infinitesimal weight, and the payoff of each agent depends on the mean-field rather that the individual states and action of all the players. 

The Poisson clock model of action, noise, and revision times allows to describe the dynamics of the mean-field as an ordinary differential equation (ODE). At time $t$, we denote the joint state-policy distribution of class $c\in [C]$ by  $\mu^c(t) \in X^c$. The collection of joint state-policy distributions of all classes is denoted by $\mu(t) :=\col(\mu^c(t), {c\in \Ccal} )\in X$. The mean-field approximation of the token economy dynamic game with evolutionary dynamics is described for all $c\in [C]$, for all $ k\in \Kcal$, and all $u\in \UDcbar$ by
\begin{equation}\label{eq:ODE_mu_ev}
	\dot{\mu}^c[k,u]  = f_{k,u}^{c,d}(\mu) + 	f_{k,u}^{c,r}(\mu),
\end{equation}
where
\begin{equation}\label{eq:fd_fr}
	\begin{split}
		\!\!\!f_{k,u}^{c,d}(\mu) &=  (\Rd\!+\!\Rno) \!\left(\sum_{k^\prime \in \Kcal}\!\phi^{c,u}(k|k^\prime)\mu^c[k^\prime,u](t)\! -\! \mu^c[k,u]\!\right)\!\!\!\!\!\!\! \\
		\!\!\!f_{k,u}^{c,r}(\mu) &= \sum_{u^\prime \in \UDcbar}  \mu^c[k,u^\prime] \rho^c_{u^\prime u}(F(\mu),\mu^c[\Kcal,\cdot]) \\
		& - \mu^c[k,u] \sum_{u^\prime \in \UDcbar}  \rho^c_{uu^\prime}(F(\mu),\mu^c[\Kcal,\cdot]).
	\end{split}\!\!\!\!\!
\end{equation}
The ODE in \eqref{eq:ODE_mu_ev} follows from balances of the proportion of mass in each class, state, and policy in an infinitesimal interval. For more details on how to obtain \eqref{eq:ODE_mu_ev} from the finite-population model see \cite[Section~IV]{PedrosoAgazziEtAl2025MFGAvg}.
Under the regularity assumptions introduced in Section~\ref{sec:model_assumptions}, one can show that the mean-field dynamics are well-posed and that they approximate the dynamics of the finite-population as $N\to \infty$.

\begin{lemma}\label{lem:mfg_approx}
	Under Assumptions~\ref{ass:affordable_prices}--\ref{ass:rev_protocol}, the mean-field model of the token economy game, characterized by \eqref{eq:ODE_mu_ev}, with initial condition $\mu(0)\in X$ has a solution $\mu(t)$ defined in $t\in [0,\infty)$, which is unique and Lipschitz continuous w.r.t.\ $\mu(0)$. Furthermore, if $\lim_{N\to \infty}\hat{\mu}(0)= \mu(0)$ almost surely, then for all $T<\infty$ $\hat{\mu}(t)$ converges in probability to $\mu(t)$ for all $t\in [0,T]$ as $N\to \infty$, i.e., $\forall T,\epsilon >0$ $\lim_{N\to \infty} \mathbb{P}(\sup_{t\in [0,T]} \| \hat{\mu}(t)-\mu(t)\|\geq\epsilon) = 0$. 
\end{lemma}
\begin{proof}
See Appendix~\ref{sec:proof_lem_mfg_approx}.
\end{proof}

\begin{mdframed}[style=callout]
	Intuitively, Lemma~\ref{lem:mfg_approx} states that, for a large population, which is reasonable in the context of a token economy, the evolution of the finite-population model, characterized by $\hat{\mu}(t)$, closely approximates that of the mean-field game, characterized by $\mu(t)$. Importantly, the mean-field model is deterministically governed by the ODE \eqref{eq:ODE_mu_ev}, which makes it analytically tractable.
\end{mdframed}

The evolutionary dynamics of the mean-field model proposed in this section resemble those of evolutionary dynamics models for population games~\cite{Sandholm2010}. However, crucially, the model in \eqref{eq:ODE_mu_ev} also captures state dynamics, which are coupled with the revision dynamics. Consequently, the extensive theory developed over the past two decades on the evolutionary game theory of population games~\cite{Sandholm2010} cannot be applied directly. Instead, we adopt arguments similar to those in recent work~\cite{PedrosoAgazziEtAl2025MFGAvg,PedrosoAgazziEtAl2025MFGAvgII} on a class of mean-field games that encompasses the model considered here.


\subsection{Equilibria}\label{sec:mf_equilibria}

In the context of game theory, solution concepts are rules or criteria that characterize reasonable outcomes of a game, such as the celebrated Nash equilibrium (NE). In this section, we analyze a solution concept of the mean-field game introduced in Section~\ref{sec:mean_field_model}. We show that under mild conditions such solution concept \emph{exists}, admits  \emph{unique resource rates}, and is the \emph{unique rest point of} \eqref{eq:ODE_mu_ev}.

We start by defining a solution concept called Mixed Stationary Nash Equilibrium (MSNE) that is proposed in \cite{PedrosoAgazziEtAl2025MFGAvg}.

\begin{definition}[MSNE]\label{def:MSNE}
	A joint state-policy distribution $\mu \in X$ is said to be a MSNE in the average payoff mean-field game, denoted by $\mu \in \MSNE(F,\phi)$, if for all $c\in [C]$ and all $u\in \UDcbar$
	\begin{align}
		&\mu^c[\Kcal,u] >0 \implies   F^c_u(\mu)  \geq F^c_v(\mu), \quad \forall v\in \UDcbar \label{eq:MSNE_u}\\
		&\mu^c[k,u] = \eta^{c,u}(k)\mu^c[\Kcal,u]  \quad \forall k\in \Kcal,   \label{eq:MSNE_s}
	\end{align}
	where $\mu^c[\Kcal,u] = \sum_{k\in \Kcal}\mu^c[k,u]$ for all $u\in \UDcbar$, and $\eta^{c,u} \in \Pcal(\Scal)$ is the stationary state distribution of the continuous-time state transition Markov chain associated with $u$, which is unique by Lemma~\ref{lem:F_def} under Assumption~\ref{ass:noise}. Moreover, for some $\epsilon\geq0$, $\mu$ is said to be an $\epsilon$-MSNE  in the average payoff mean-field game, denoted by $\mu \in \MSNE_\epsilon(F,\phi)$, if  for all $c\in [C]$ and all $u\in \UDcbar$
		\begin{align}
		&\mu^c[\Kcal,u] > \epsilon \implies   F^c_u(\mu)  \geq F^c_v(\mu) -\epsilon, \quad \forall v\in \UDcbar \label{eq:epsMSNE_u}\\
		&|\mu^c[k,u] - \eta^{c,u}(k)\mu^c[\Kcal,u]|\leq \epsilon  \quad \forall k\in \Kcal. \label{eq:epsMSNE_s}
\end{align}
\end{definition}

The MSNE is a simple and intuitive notion of equilibrium. It is a condition in which each agent adopts a deterministic policy in steady-state and has no incentive to switch to any other available deterministic policy. Moreover, an $\epsilon$-MSNE is a relaxed version of this condition, where each agent satisfies the requirements of an MSNE up to an accuracy level $\epsilon$. Crucially, we show in the following result that an MSNE exists and that the corresponding equilibrium resource flows are unique.

\begin{theorem}\label{th:unique_MSNE}
	Under Assumptions~\ref{ass:affordable_prices}-\ref{ass:noise}, $\MSNE(F,\phi)$ is non-empty, compact, and convex; and the equilibrium resource flows are unique, i.e., $\sigma(\mu) = \sigma(\nu)$ for all $\mu,\nu \in \MSNE(F,\phi)$. Moreover, for all $\delta > 0$ there exists $\epsilon> 0$ such that $\|\sigma(\mu)-\sigma(\nu)\| \leq\delta$ for all $\mu \in \MSNE_{\epsilon}(F,\phi)$ and all $\nu \in \MSNE(F,\phi)$.
\end{theorem} 
\begin{proof}
	See Appendix~\ref{sec:proof_th_unique_MSNE}.
\end{proof}

\begin{mdframed}[style= callout]
	The uniqueness of the equilibrium resource utilization flows established in Theorem~\ref{th:unique_MSNE} is a remarkable property from a prescriptive perspective. Specifically, if the toll maps are designed so that an equilibrium with optimal resource utilization rates (according to criterion~(ii) of Problem~\ref{prob:design}) exists, then all other equilibria necessarily achieve the same optimal resource utilization rates.
\end{mdframed}


\vspace{-0.3cm}

\subsection{Stability of Equilibria}\label{sec:mf_stability}

In this section, we show that the evolutionary dynamics of the mean-field game characterized by \eqref{eq:ODE_mu_ev} converge close to the set of MSNE in a regime where the time scale separation of the action and revision dynamics is sufficiently large, which is reasonable since $\Rr \ll \Rd$.


\begin{theorem}\label{th:global_stable}
	Under Assumptions~\ref{ass:affordable_prices}--\ref{ass:rev_protocol}, consider an imitative payoff or pairwise comparison revision protocol $\rho^c$ for each class $c\in [C]$. Then, for all $\epsilon>0$ and all $\mu(0)\in \mathrm{int}(X)$ there exist $\delta,T >0$ such that for all $\Rr/(\Rd+\Rno) < \delta$ trajectories of \eqref{eq:ODE_mu_ev} satisfy $\mu(t) \in \MSNE_\epsilon(F,\phi)\;\;\; \forall t \geq T$.
\end{theorem}
\begin{proof}
	See Appendix~\ref{sec:proof_th_global_stable}.
\end{proof}

\begin{mdframed}[style= callout]
	The global convergence guarantee established in Theorem~\ref{th:global_stable} is another remarkable property from a prescriptive perspective. Specifically, it states that, for the entire classes of imitative and pairwise comparison revision protocols, the mean-field trajectory $\mu(t)$, which by Lemma~\ref{lem:mfg_approx} provides a good approximation of the finite-population token economy, enters a set of $\epsilon$-MSNE in finite time.
\end{mdframed}

\begin{remark}
	Notice that the convergence result of Theorem~\ref{th:global_stable} only holds if the initial condition is in the interior of $X$, i.e.,  $\mu(0)\in \mathrm{int}(X)$. That is because if no player initially chooses a policy $u$, then policy $u$ will not ever be chosen when revisions rely on the imitation of other players. Nevertheless, inevitable perturbations of the revision protocol in realistic applications make trajectories move away from such degenerate conditions. Furthermore, if one considers instead that each class plays an \emph{hybrid} protocol, i.e., if at each revision ring each agent plays an imitative or pairwise comparison randomly with a non-degenerate proportion, then the result of Theorem~\ref{th:global_stable} holds for any initial condition  $\mu(0)\in X$. It is straightforward to show this from the proof of Theorem~\ref{th:global_stable} and the properties of hybrid protocols in \cite[Chapter~5.7]{Sandholm2010}.
\end{remark}

\section{Prescriptive Analysis}\label{sec:design}

In this section, we focus on the design of the toll maps with the goal of proposing a solution to Problem~\ref{prob:design}. For that, we make extensive use of the descriptive analysis of Section~\ref{sec:descriptive}. First, we show that for any toll maps that satisfy Assumption~\ref{ass:affordable_prices}, the intra-class fairness condition in \eqref{eq:intra_class_fairness} is satisfied at a MSNE.


\begin{theorem}\label{th:opt_intra_class_fair}
For any $c\in [C]$, let $i,j$ represent independent samples drawn from $\Ccal^c$ uniformly at random.
Under Assumptions~\ref{ass:affordable_prices}--\ref{ass:noise}, consider any toll map collection $\tau$, then for any $\delta\in (0,1)$ there exists $\epsilon >0$ such that $\mathbb{P}(|J^i(\mu) - J^j(\mu)|\leq \delta) \geq (1-\delta)^2$ for all $\mu\in \MSNE_\epsilon(F,\phi)$. Moreover, the result also holds for $\delta = 0$ with $\epsilon = 0$.
\end{theorem}
\begin{proof}
	See Appendix~\ref{sec:proof_th_opt_intra_class_fair}.
\end{proof}

\begin{mdframed}[style= callout]
	It follows from Theorem~\ref{th:opt_intra_class_fair} that, for any toll map, the intra-class fairness condition  \eqref{eq:intra_class_fairness} is satisfied with arbitrary precision in a set of $\epsilon$-MSNE, which is reached in finite-time from any initial condition from Theorem~\ref{th:global_stable}.
\end{mdframed}

Second, one may compute a maximizer of the chosen combination of efficiency and fairness in line with criterion~(ii) of Problem~\ref{prob:design} as $f^\star \in $
\begin{equation}\label{eq:sys_opt}
	\begin{aligned}
		\!\!\!\underset{f^c_a,\; c\in[C], a\in \Acal^c}{\argmax}\;\; &g(\Jeff(\sigma),\JF(\bar w)) \\
		\mathrm{s.t.}\quad  \quad \;
		&  \sigma_r = \sum_{c\in[C]}\sum_{a\in \Acal^c_{[\Rsf]}(r)} f^c_a, \; r\in[\Rsf] \\
		&  \bar w_c =  \sum_{a\in \Acal^c} f^c_{a}w(a,\sigma)/(m^c\Rd), \;  c\in[C]\\
		&  \sum_{a\in \Acal^c} f^c_a  = m^c\Rd , \quad c\in[C] \\
		& f^c_a \geq 0, \quad  c\in[C], \;a\in \Acal^c.
	\end{aligned}
\end{equation}
A solution $f^\star$ to \eqref{eq:sys_opt} is characterized by utilization flows of each action $a\in \Acal^c$ and each class $c\in[C]$ denoted by $f^{c\star}_a$; utilization flows for each resource $r\in [\Rsf]$ denoted by $\sigma^\star_r$; and average payoffs for each class $c\in [C]$, denoted by $\bar w_c^\star$. Solving \eqref{eq:sys_opt} efficiently may be challenging depending on the chosen function $g$ under Assumption~\ref{ass:J}. Finding a solution is outside of the scope of this paper, so, henceforth, we assume that a solution $f^\star$ is known.

\begin{mdframed}[style= callout]
From Theorem~\ref{th:unique_MSNE}, all MSNE have the same equilibrium resource flows. Therefore, if one designs tolls $\tau$ such that there is a MSNE that has optimal resource flows $\sigma^\star$, then every MSNE also has optimal resource flows $\sigma^\star$. Furthermore, during the design procedure one can consider the single-stage rewards of every action to be fixed and equal to the single-stage rewards at the optimal resource flows denoted by $w^\star_a = w(a,\sigma^\star)$ for all $a\in \Acal$. Then, if there is a MSNE with optimal resource flows for constant single-stage rewards, it is also a MSNE for the original game.
\end{mdframed}

We start by considering, for any class $c\in [C]$, an auxiliary optimization problem that characterizes the best-case payoff under fixed single-stage rewards at $\sigma^\star$ and where continuous tolls $\tilde{\tau}_a^c$ with $c\in [C]$ and $a\in \Acal^c$ are allowed  
\begin{equation}\label{eq:class_opt}
	\begin{aligned}
		\underset{\nu^c_a, a\in \Acal^c}{\max}\quad & { \sum\nolimits_{a\in \Acal^c} w_a^\star \nu^c_a} \\
		\mathrm{s.t.}\quad  \quad 
		& {  m^c\Rno- \sum\nolimits_{a\in \Acal^c} \tilde{\tau}_a^c \nu^c_a \geq 0}\\
		& {  \sum\nolimits_{a\in \Acal^c} \nu^c_a  = m^c\Rd}\\
		& \nu^c_a \geq 0, \quad  a\in \Acal^c.
	\end{aligned}
\end{equation}
In \eqref{eq:class_opt}, the decision variables $\nu^c_a \geq 0$ with $a\in \Acal^c$ are the action utilization rates of class $c$ and the objective is to maximize the average payoff under fixed single-stage rewards at $\sigma^\star$. The first and second constraints enforce that the average rate of variation of the amount of tokens must not be negative (agents cannot be depleting on average their token wallets), and that the class flow is conserved, respectively. Furthermore, \eqref{eq:class_opt} is a Linear Program (LP), which allows the use of strong duality theorems for its analysis.

The LP \eqref{eq:class_opt} lies at the center of the proposed design procedure. First, we design continuous tolls $\tilde{\tau}$ such that a solution to \eqref{eq:class_opt} is the solution $f^{c\star}_a$ with $a\in \Acal^c$ to \eqref{eq:sys_opt} for the action flows of each class $c\in[C]$. Second, we show that there exists a distribution of policies for each class $c$ that prescribes action flows whose average payoff is arbitrarily close to the optimal value of \eqref{eq:class_opt} and that no other policy can achieve a better payoff. Hence, such policy distributions for each class characterize a $\epsilon$-MSNE. Third, putting the previous two steps together for arbitrary $\epsilon >0$, we can obtain an expression for integer tolls $\tau$ that prescribe an $\epsilon$-MSNE with optimal resource flows $\sigma^\star$ and optimal average payoffs for each class $\bar w^\star$. The three steps outlined above are presented in the three following results, starting with the design of continuous tolls $\tilde{\tau}$ as follows.

\begin{lemma}\label{lem:opt_prices}
	Consider any class $c\in[C]$, any $a\in \Acal^c$, and let $f^{c\star}_a$ be a solution to \eqref{eq:sys_opt}. Then, for tolls $\tilde{\tau}^c_a(\alpha)$ given by
	\begin{equation}\label{eq:opt_prices_cont}
		\begin{cases}
			\tilde{\tau}^c_a(\alpha) = \Rno/\Rd +\alpha (w^{\star}_a -\bar{w}_c^\star), &\quad f^{c^\star}_a >0\\
			\tilde{\tau}^c_a(\alpha) \geq  \Rno/\Rd +\alpha (w^{\star}_a -\bar{w}_c^\star), &\quad f^{c^\star}_a =0,
		\end{cases}		
	\end{equation}
	where $\alpha>0$ is any scaling factor, $f^{c\star}_a$ with $a\in \Acal^c$ is a maximizer of \eqref{eq:class_opt}.
\end{lemma}
\begin{proof}
	See Appendix~\ref{sec:proof_lem_opt_prices}.
\end{proof}

\begin{lemma}\label{lem:opt_pol}
	Under Assumptions~\ref{ass:senseless_policies} and~\ref{ass:noise}, consider a discretization $\tau^c_a(\alpha) = \mathrm{round}(\tilde{\tau}^c_a(\alpha))$ for all $c\in [C]$ and all $a\in \Acal^c$ of the continuous tolls in \eqref{eq:opt_prices_cont} in Lemma~\ref{lem:opt_prices} with scaling factor $\alpha>0$. Then, for all $\epsilon>0$ there is $\alpha^\star>0$ and $\bar{k}^\star\in \N$  such that for all $c\in [C]$, all $\alpha\geq\alpha^\star$, and all $\bar{k} \geq  \bar{k}^\star$ there exists $\mu^c\in X^c$ such that for all $u\in \UDcbar$
	\begin{equation}\label{eq:opt_pol_cond1}
		\begin{split}
		\!\!\!\mu^c[\Kcal,u] >0 & \implies \sum_{k\in \Kcal} \sum_{a\in \Acal^c} \eta^{c,u}(k)u(a|k)w^\star_a \\ & \!\!\!\geq \sum_{k\in \Kcal} \sum_{a\in \Acal^c} \!\eta^{c,v}(k)v(a|k)w^\star_a -\epsilon,  \quad \!\! \forall  v\in \UDcbar
	 \end{split}
	\end{equation}
and for all $a\in \Acal^c$
	\begin{equation}\label{eq:opt_pol_cond2}
		\Big|f^{c\star}_a-   \sum\nolimits_{u\in \UDcbar} \sum\nolimits_{k\in \Kcal}\mu^c[k,u]u(a|k)\Big| \leq \epsilon.
	\end{equation}
\end{lemma}
\begin{proof}
	See Appendix~\ref{sec:proof_lem_opt_pol}.
\end{proof}

Intuitively, for any $\epsilon >0$ and any class $c$, there is a joint state-policy distribution $\mu^c\in X^c$ that places mass on $\epsilon$-optimal policies and whose action flows differ at most by $\epsilon$ w.r.t\  $f_a^{c\star}$ with $a\in \Acal^c$. The main result of this paper follows from these two lemmas.

\begin{theorem}\label{th:opt_prices_MSNE}
	Under Assumptions~\ref{ass:cont}-\ref{ass:noise}, for each $c\in [C]$ consider a discretization $\tau^c_a(\alpha) = \mathrm{round}(\tilde{\tau}^c_a(\alpha))$ for all $a\in \Acal^c$ of the continuous tolls in \eqref{eq:opt_prices_cont} in Lemma~\ref{lem:opt_prices} with scaling factor $\alpha>0$. Then, for all $\delta>0$ there is $\epsilon>0$, $\alpha^\star>0$ and $\bar{k}^\star\in \N$ such that $|\bar \sigma_r(\mu) -\sigma^\star_c| \leq \delta$ and $|\bar w_c(\mu) -\bar w^\star_c| \leq \delta$ for all $r\in \Rcal$, all $c\in \Ccal$, all $\mu\in \MSNE_\epsilon(F,\phi)$, all $\alpha\geq\alpha^\star$, and all $\bar{k} \geq  \bar{k}^\star$.
\end{theorem}
\begin{proof}
	See Appendix~\ref{sec:proof_th_opt_prices_MSNE}.
\end{proof}

\begin{mdframed}[style= callout]
	One can now bring the three main results of this paper together. Theorem~\ref{th:opt_intra_class_fair} establishes that for any toll map the intra-class fairness condition is satisfied with arbitrary accuracy in a set of $\epsilon_1$-MSNE. Theorem~\ref{th:opt_prices_MSNE} provides a closed-form expression for integer tolls for which an optimal combination of efficiency and fairness is achieved with arbitrary accuracy in a set of $\epsilon_2$-MSNE.
	By Theorem~\ref{th:global_stable}, the state of the mean-field token economy model (which is a good approximation of the finite-population model by Lemma~\ref{lem:mfg_approx}) enters a set of $\epsilon$-MSNE in finite time with $\epsilon = \min(\epsilon_1,\epsilon_2)$ for a sufficiently large time-scale separation. Hence, both goals of Problem~\ref{prob:design} are achieved.
\end{mdframed}



Since Theorem~\ref{th:global_stable} guarantees global convergence to the set of MSNE, the framework naturally extends to settings where some parameters vary slowly over time; one only needs to recompute the tolls as the parameters (e.g., class masses, resource availability, or fairness metrics) change.

\vspace{-0.2cm}

\section{Illustrative Application}\label{sec:illustrative_example}

We consider a network congestion game on the illustrative traffic network of Sioux Falls, USA \cite{TNRCT2025}. We consider a population of approximately $2\times 10^5$ agents, $76$ resources which are the roads of the traffic network, and $117$ classes each corresponding to agents that desire to travel between each origin-destination pair. The reward function of each resource is the symmetric of the travel time on that road link, which decreases monotonically as more agents use that link according to a quartic BPR model \cite{BPR1964}. We choose an average reward efficiency metric, i.e., $\Jeff(\sigma) = \sum_{r\in [\Rsf]}\sigma_r w_r(\sigma_r)$, and a minimum class-average reward inter-class fairness metric, i.e., $\JF(\bar{w}) = \min_{c\in[C]} \bar w_c$. We desire to maximize the performance metric $\JF(\bar{w}) + \gamma\Jeff(\sigma)$, where $\gamma = 0.1$ is a small regularization weight. Some simulation results are depicted in Fig.~\ref{fig:sim}. From Figs.~\ref{fig:w_bar} and~\ref{fig:cost} one can immediately see that, as expected, the average payoff of the classes converge close to the optimal average payoffs and the performance metric also converges to the optimum. The gap between the limit of the simulation results and the corresponding optima can be made arbitrarily close to zero by increasing $N$, $\alpha^\star$, and $\bar{k}^\star$.
Interestingly, Fig.~\ref{fig:cost} also depicts the asymptotic value of the performance metric if no incentive scheme were implemented (i.e., zero tolls for every action), which is roughly $13\%$ worse. Fig.~\ref{fig:tokens} shows the token distribution of one of the classes, which also converges to a stationary distribution. Fig.~\ref{fig:actions} shows the joint distribution of actions and tokens at the end of the simulation, which corresponds to the combination of different policies with identical reward rather than a single policy. For the sake of conciseness, the details of the simulation were omitted. Nonetheless, the MATLAB implementation used to carry out the design and simulation is available in an open-source repository at \weblink{https://github.com/fish-tue/token-economy}.



\begin{figure}[ht]
	\centering
	\begin{subfigure}[b]{\linewidth}
		\centering
		\includegraphics[width=0.82\linewidth]{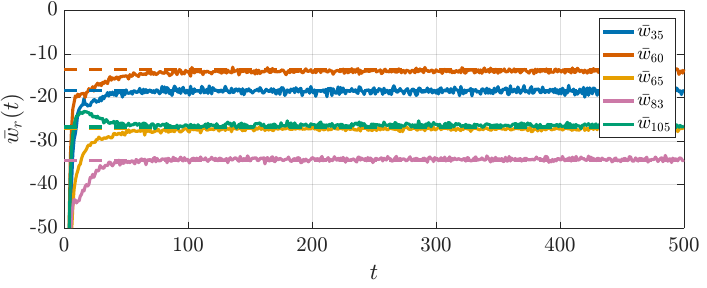} 
		\vspace{-0.3cm}
		\caption{Evolution of the resource utilization flows.}
		\label{fig:w_bar}
	\end{subfigure}
	
	\begin{subfigure}[b]{\linewidth}
		\centering
		\includegraphics[width=0.82\linewidth]{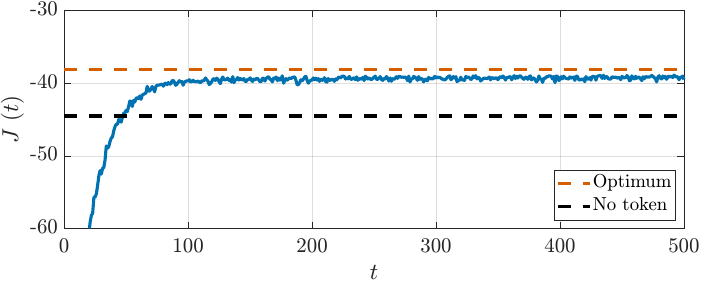}
		\vspace{-0.3cm}
		\caption{Evolution of the performance metric.}
		\label{fig:cost}
	\end{subfigure}
	\vspace{0.1cm} 
	\begin{subfigure}[b]{\linewidth}
		\centering
		\includegraphics[width=0.82\linewidth]{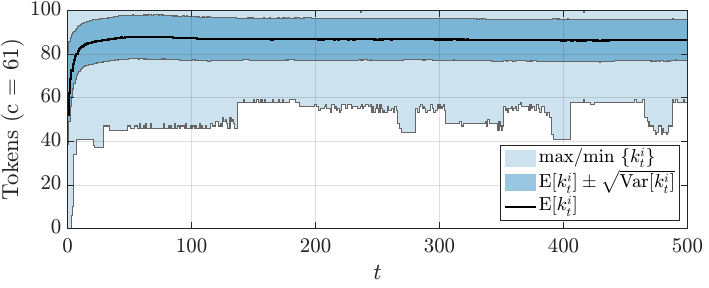}
		\vspace{-0.3cm}
		\caption{Evolution of the token distribution of class $61$.}
		\label{fig:tokens}
	\end{subfigure}%
	\hfill
	\begin{subfigure}[b]{\linewidth}
		\centering
		\includegraphics[width=0.82\linewidth]{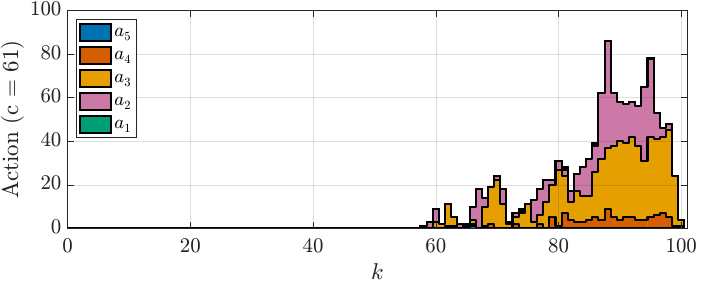}
		\vspace{-0.3cm}
		\caption{Token-action distribution of class $61$ at the end of the simulation.}
		\label{fig:actions}
	\end{subfigure}
	\caption{Simulation of the token economy for the illustrative traffic network.}
	\label{fig:sim}
	\vspace{-0.2cm}
\end{figure}


\section{Conclusion}\label{sec:conclusion}
\vspace{-0.1cm}
In this paper, we proposed a token economy framework for dynamic resource allocation in congestion games. We considered three performance metrics: (i) intra-class fairness, ensuring that users within the same class experience the same average reward over time regardless of individual factors such as wealth; (ii) inter-class fairness, which concerns how the average reward is distributed across classes (e.g., favoring classes that are worse off); and (iii) efficiency, which seeks to maximize the overall utility of the resources (e.g., maximizing the average reward or minimizing environmental impact).
The design problem for the token economy amounts to finding tolls that guarantee intra-class fairness and optimize a desired trade-off between inter-class fairness and efficiency. We showed that the token economy can be modeled as a dynamic game with many players and is well approximated by a mean-field model. A Nash-like equilibrium for this mean-field approximation was shown to exist and to be essentially unique (i.e., the equilibrium resource utilization flows are unique). Moreover, the state of the token economy converges to a neighborhood of the equilibria from any initial condition.
We further established that intra-class fairness is achieved for any choice of tolls, and we derived a closed-form expression for integer tolls that induce an equilibrium at the optimal trade-off between efficiency and inter-class fairness. Finally, numerical simulations illustrated and validated the theoretical results.

\vspace{-0.1cm}
\appendix
\vspace{-0.2cm}
\section{Proofs of Section~\ref{sec:model}}

\subsection{Proof of Lemma~\ref{lem:F_def}}\label{sec:proof_lem_F_def}

First, we show by contradiction that $\phi^{c,u}$, defined in \eqref{eq:equiv_jump_chain_noise}, has a unique recurrent communicating class. Indeed, assume that there are multiple recurrent communicating classes in $\phi^{c,u}$. Consider two distinct communicating classes $\Kcal_1$ and $\Kcal_2$ and let $\Kcal_1$ be the class with the lowest maximum token amount of the two. Due to the regularization noise in $\phi^{c,u}$ under Assumption~\ref{ass:noise}, there is a nonnull probability of transition from every state of $\Kcal_1$ to a state of $\Kcal_2$ in a finite number of jumps, which is a contradiction since then $\Kcal_1$ cannot be a recurrent communicating class. Second, from a simple extension of \cite[Theorem~3.5.2]{Norris1997} (for details on the extension see \cite[Lemma~2.2]{PedrosoAgazziEtAl2025MFGAvg}) it follows that continuous-time Markov chains whose embedded Markov chains have one and only one recurrent class admit a unique stationary distribution, which shows the first statement of the lemma immediately. Third, we prove the second claim. For any $c\in [C]$ consider any agent $i\in \Ccal_c$ that uses any policy $u\in \Ucal_D^c$. Under Assumption~\ref{ass:noise}, it follows from the analysis above that $\phi^{c,u}$ has one and only one recurrent communicating class. As a result, from a simple extension of \cite[Theorem~3.8.1]{Norris1997} (for details on the extension see \cite[Lemma~2.2]{PedrosoAgazziEtAl2025MFGAvg}), the state distribution of agent $i$ converges almost surely to $\eta^{c,u}$, from any initial condition, as $t\to \infty$. Additionally, by Assumption~\ref{ass:cont}, the single stage reward is continuous and takes arguments on a compact set, therefore it is bounded. Hence, for a fixed $\hat{\sigma}$, one can use the Dominated Convergence Theorem \cite[Theorem~9.1.2]{Rosenthal2006} to exchange the limit and expected value in \eqref{eq:def_Ji} to express the long-time average reward as \eqref{eq:Ji_avg}.\hfill$\qedsymbol$%
%
\section{Proofs of Section~\ref{sec:descriptive}}

\vspace{-0.1cm}
\subsection{Proof of Lemma~\ref{lem:mfg_approx}}\label{sec:proof_lem_mfg_approx}
Notice that the vector field of the ODE characterized by \eqref{eq:ODE_mu_ev} is in the tangent space of $X$ under Assumption~\ref{ass:rev_protocol} and is Lipschitz continuous under Assumptions~\ref{ass:cont} and~\ref{ass:rev_protocol}. Existence and uniqueness of a solution to \eqref{eq:ODE_mu_ev} follows from an extension to compact convex spaces of the Picard-Lindel\"{o}f Theorem \cite[Theorem~5.7]{Smirnov2002} and Lipschitz continuity follows from Gr\"onwall's Inequality \cite[Theorem~4.A.3]{Sandholm2010}. Finally, convergence in probability follows form a direct application of Kurtz's Theorem \cite[Theorems~10.2.1 and~10.2.3]{Sandholm2010}.\hfill$\qedsymbol$%

\subsection{Proof of Theorem~\ref{th:unique_MSNE}}\label{sec:proof_th_unique_MSNE}

First, we resort to the notion of steady-state game to obtain a simpler characterization of the set of MSNE. The steady-state game is characterized by the particularization of the payoff map $F$ when the state dynamics are stationary. Formally, let $x^c \in X^c_{\Ucal_D}$ denote a distribution of policies for class $c\in [C]$. Denote also the collection of policy distributions of all populations by $x:= \col(x^c, c\in [C])\in X_{\Ucal_D}$. Notice that, given $x\in X_{\Ucal_D}$, the state-policy distribution in steady-state associated with $x$ is denoted by $\bar{\mu}(x) \in X$ and characterized as $\bar{\mu}^c(x)[k,u] = x^c[u]\eta^{c,u}(k)$  for all $c\in [C]$, all $k\in \Kcal$, and all $u\in \UDcbar$, where $\eta^{c,u} \in \Pcal(\Scal)$ is the stationary state distribution of the continuous-time state transition Markov chain associated with $u$, which is unique by Lemma~\ref{lem:F_def} under Assumption~\ref{ass:noise}.
The steady-state payoff map of class $c\in [C]$ is a map $\Fcal^c : \bar{X} \to \R^{n^c}$ defined as $\Fcal^c(x) = F^c(\bar{\mu}(x))$. We also denote the collection of maps of all classes by $\Fcal := \col(\Fcal^c, c\in [C])$.

The steady-state game admits a standard definition of Nash equilibrium. Formally, we say that $x\in X_{\Ucal_D}$ is a Nash equilibrium of the steady-state game, which is denoted by $x\in \NE(\Fcal)$, if for all $c\in [C]$ and all $u\in \UDcbar$ \;$x^c[u]>0 \implies \Fcal^c_u(x)\geq \Fcal^c_v(x), \; \forall v\in \UDcbar$. Remarkably, comparing the definitions of MSNE of the mean-field game and the definition of NE of the steady-state game, one can immediately show the equivalence between both as follows.

\begin{proposition}\label{prop:NEss_iff_MSNE} 
	Consider $x\in X_{\Ucal_D}$ and $\mu = \bar{\mu}(x)$. Then $x\in  \NE(\Fcal)  \iff \mu \in \MSNE(F,\phi)$.
\end{proposition}

Second, using Proposition~\ref{prop:NEss_iff_MSNE}, one can analyze the Nash equilibria of the steady-state game to draw conclusions about the set of MSNE. Specifically, we resort to the fact that the steady-state game is a full potential game and to well-known results to characterize $\NE(\Fcal)$ as shown in the following result.

\begin{proposition}\label{prop:ss_is_potential}
	The steady-state game characterized by $\Fcal$ is a full potential game, i.e., $\Fcal$ admits a domain extension to $\Rnn^n$ and there exists a continuously differentiable function $U:\Rnn^n \to \R$ such that $ \Fcal = \nabla U$. Furthermore, $U$ can be expressed as $U(x) = \tilde{U}(\sigma(\bar\mu(x))), \forall x\in \Rnn^n$, where $\tilde U : \Rnn^R \to \R$ is strongly convex. Moreover, $\NE(\Fcal)$ is non-empty, compact, and convex, and the equilibrium resource flows are unique, i.e., $\sigma(\bar{\mu}(x)) = \sigma(\bar{\mu}(y))$ for all $x,y \in \NE(\Fcal)$.
\end{proposition}
\begin{proof}
	Expanding $\Fcal^c_u$ with \eqref{eq:def_F} yields
	\begin{equation}\label{eq:F_x}
		\Fcal_u^c(x) = \sum_{k\in \Kcal} \sum_{a \in \Acal^c} \sum_{r\in a} \eta^{c,u}(k)u(a|k)w_r(\sigma_r(\bar{\mu}(x))).
	\end{equation}
	for all $c\in [C]$ and all $u\in \UDcbar$. Notice that for all $c\in [C]$ and all $u\in \UDcbar$, $\Fcal_u^c(x)$ in \eqref{eq:F_x} admits a domain extension to $\Rnn^{n}$ since the resource reward function $w_r$ are defined in $\Rnn$ and the maps $\sigma_r(\cdot)$ for $r\in [\Rsf]$ and $\bar{\mu}(\cdot)$ are consistent with the domain extension. Henceforth, we consider the domain extension of $\Fcal^c : \Rnn^n \to \R^{n^c}$.
	First, define $U(x)  =  \frac{1}{\Rd}\sum_{r\in [\Rsf]} \int_0^{\sigma_r(\bar{\mu}(x))}w_r(s) \dint s$
	as the candidate potential function. After algebraic manipulation, which we omit for the sake of conciseness, it follows that $\partial U(x)/\partial x^c[u] = \Fcal_u^c(x)$ for any $c\in [C]$ and any $u\in \UDcbar$. Hence, $\Fcal$ is a full potential game. %
	Second, we establish two concavity properties of $U$. Define  $\tilde U(\sigma) = (1/\Rd)\sum_{r\in [\Rsf]} \int_0^{\sigma_r}w_r(s) \dint s$ which can be written as $\tilde U(\sigma) = (1/\Rd)\sum_{r\in [\Rsf]} u_r(\sigma_r)$, where $u_r(z) = \int_0^zw_r(s)\dint s$ for all $r\in [\Rsf]$. The first derivative of $u_r(z)$ exists and is given by $\dint u_r/\dint z = w_r(z)$, which by Assumption~\ref{ass:cont} is strictly decreasing for all $r\in [\Rsf]$. Hence, $u_r(z)$ is strictly concave for all $r\in [\Rsf]$ and so is  $\tilde U(\sigma)$. Since the composition with an affine map preserves concavity and the map $\sigma(\bar{\mu}(\cdot))$ is linear, it follows that $U(x)$ is concave.
	Third, since $U(x)$ is concave, it is well known that $\NE(\Fcal)$ corresponds to the set of maximizers of $U(x)$ in $X_{\Ucal_D}$ \cite[Corollary~3.1.4]{Sandholm2010}. Since $X_{\Ucal_D}$ is compact and $U(x)$ is continuous it follows that $\NE(\Fcal)$ is non-empty, compact, and convex. Furthermore, since $\tilde U$ is strictly concave the maximizer of $\tilde U$ is unique and, therefore, the equilibrium resource flows are unique. Notice that one could state a stronger concavity property for $\tilde U$. Indeed, since $w_r(z)$ is uniformly decreasing, it follows immediately that $\tilde U$ is strongly concave.
\end{proof}

Together with the relation between $\NE(\Fcal)$ and $\MSNE(F,\phi)$ in Proposition~\ref{prop:NEss_iff_MSNE}, Proposition~\ref{prop:ss_is_potential} shows the first statement.
Consider now a generic $\epsilon>0$ and any $\mu \in \MSNE_\epsilon(F,\phi)$. Define $\zeta \in X$ as $\zeta^c[k,u] = \mu^c[\Kcal,u]\eta^{c,u}(k)$ for all $c\in[C]$, all $k\in \Kcal$, and all $u\in \UDcbar$. The distribution $\zeta$ has two important properties: (i)~$\zeta^c[\Kcal,u] = \mu^c[\Kcal,u]$ for all $c\in [C]$ and all $u\in \UDcbar$; and (ii)~by the definition of $\epsilon$-MSNE in Definition~\ref{def:MSNE} for all $c\in [C]$, all $k\in \Kcal$, and all $u\in \UDcbar$
\begin{equation}\label{eq:zeta_mu_error}
	|\zeta[k,u]-\mu[k,u]| \leq |\mu^c[\Kcal,u]\eta^{c,u}(k) -\mu^c[k,u] | \leq \epsilon.
\end{equation}
Define additionally $z\in X_{\Ucal_D}$ as $z^c[u] = \zeta^c[\Kcal,u]$ for all $c\in [C]$ and all $u\in \UDcbar$, which satisfies: (i)~$\bar\mu(z) = \zeta$ by construction; and (ii)~$\|\mu-\bar\mu(z)\|\leq l_1\epsilon$ for some $l_1>0$ by \eqref{eq:zeta_mu_error}. Moreover, for all $c\in [C]$ and all $u\in \UDcbar$ if $z^c[u] >\epsilon$ then $\mu^c[\Kcal,u]>\epsilon$ and by the definition of $\epsilon$-MSNE in Definition~\ref{def:MSNE} $F^c_u(\mu)\geq F^c_v(\mu)-\epsilon$ for all $v\in \UDcbar$. The payoff is globally Lipschitz continuous by Assumption~\ref{ass:cont}, and $\|\mu-\bar\mu(z)\|\leq l_1\epsilon$, thus there is $l_2>0$ such that $|F^c_v(\mu)-\Fcal_v^c(z)| \leq l_2\epsilon$ for all $c\in [C]$, all $v\in \UDcbar$, and any choice of $\mu \in \MSNE_\epsilon(F,\phi)$. One concludes that for all $c\in [C]$ and all $u\in \UDcbar$
\begin{equation}\label{eq:NE_z}
	z^c[u] > \epsilon \implies \Fcal_u^c(z) \geq  \Fcal_v^c(z) -l_3\epsilon\; \forall v\in \UDcbar,
\end{equation}
where $l_3:= 1 +2l_2$. Define $y\in X_{\Ucal_D}$ as $y\in \argmin_{y\in X_{\Ucal_D}} \|y-z\| \; \text{s.t.}\; z^c[u] \leq \epsilon \implies y^c[u] = 0 \; \forall c\in [C]\; \forall u\in \UDcbar$. The policy distribution $y$ satisfies: (i)~by construction there is $l_4>0$ such that  $\|z-y\|\leq l_4\epsilon$; and (ii)~$y^c[u]>0 \implies z^c[u]> \epsilon$.  Again, since the payoff is globally Lipschitz continuous and $\|z-y\|\leq l_4\epsilon$, it follows that there is $l_5 >0$ such that $|\Fcal^c_v(z)-\Fcal_v^c(y)| \leq l_5\epsilon$ for all $c\in [C]$ and all $v\in \UDcbar$. One concludes from \eqref{eq:NE_z} that for all $c\in [C]$ and all $u\in \UDcbar$
\begin{equation}\label{eq:NE_y}
	y^c[u] > 0 \implies \Fcal_u^c(y) \geq  \Fcal_v^c(y) -l_6\epsilon\; \forall v\in \UDcbar,
\end{equation}
where $l_6:= l_3+2l_5$. Furthermore, since $\|\mu-\bar\mu(z)\|\leq l_1\epsilon$ and $\|z-y\|\leq l_4\epsilon$, it follows there exists $l_7>0$ such that $\|\mu-\bar\mu(y)\| \leq l_7\epsilon$. Notice that for any policy distribution $x\in X_{\Ucal_D}$, using \eqref{eq:NE_y}, for all $c\in [C]$
\begin{equation*}
	\begin{split}
			\!\Fcal(y)^{\!\top}(x\!-\!y) &=\!\! \sum_{c \in [C]}\sum_{u\in \UDcbar} \!x^c[u]\Fcal_u^c(y) -  \! \sum_{c \in [C]}\sum_{u\in \UDcbar} \!y^c[u]\Fcal_u^c(y)\\
			&\leq  \sum_{c \in [C]}m^c\Fcal^{c\star} -  \sum_{c \in [C]}m^c(\Fcal^{c\star}  - l_6\epsilon) = l_6\epsilon,
	\end{split}
\end{equation*}
where $\Fcal^{c\star}  = \max_{v\in \UDcbar} \Fcal^c_v(y)$.
Particularizing for a $x^\star\in \NE(\Fcal)$ yields $\Fcal(y)^\top(x^\star-y) \leq   l_6\epsilon$, which using the potential function of Proposition~\ref{prop:ss_is_potential} can be written as $\nabla U(y)^\top(x^\star-y) \leq   l_6\epsilon$. Therefore, since the potential $U$ is concave by the proof of Proposition~\ref{prop:ss_is_potential}, it follows that $U(x^\star)-U(y) \leq \nabla U(y)^\top(x^\star-y) \leq   l_6\epsilon$, which can be rewritten as
\begin{equation}\label{eq:eps_NE_potential}
	\tilde U(\sigma(\bar\mu(x^\star)))-\tilde U(\sigma(\bar \mu (y))) \leq l_6\epsilon.
\end{equation}
Since $\tilde U$ is strongly concave by Proposition~\ref{prop:ss_is_potential} and the fact that $\sigma(\bar\mu(x^\star))$ is the unique maximizer of $\tilde U$, it follows from \eqref{eq:eps_NE_potential} that there exists $l_8>0$ such that $\|\sigma(\bar\mu(x^\star)) - \sigma(\bar \mu (y))\| \leq l_8 \epsilon$. Since  $\|\mu-\bar\mu(y)\| \leq l_7\epsilon$ and the fact that the map $\sigma(\cdot)$ is linear one concludes that there exists $l_9>0$ such that $\|\sigma(\bar\mu(x^\star)) - \sigma(\mu)\| \leq l_9\epsilon$. By the first statement of the result, resource utilization rates are unique for all MSNE so we conclude that  $\|\sigma(\mu^\star) - \sigma(\mu)\| \leq l_9\epsilon$ for all $\mu^\star \in \MSNE(F,\phi)$.
Hence, for all $\delta>0$ one can choose $\epsilon = \delta/l_9$ such that $\|\sigma(\mu^\star) - \sigma(\mu)\| \leq \delta$ for all $\mu \in \MSNE_\epsilon(F,\phi)$ and all $\mu^\star \in \MSNE(F,\phi)$. \hfill$\qedsymbol$%

\subsection{Proof of Theorem~\ref{th:global_stable}}\label{sec:proof_th_global_stable}
Since the steady-state game $\Fcal$ is a full potential game by Proposition~\ref{prop:ss_is_potential}, then we are in the conditions of \cite[Corollary~4]{PedrosoAgazziEtAl2025MFGAvgII}, which shows that for all $B>0$ and all $\mu(0)\in \mathrm{int}(X)$ there exist $\delta,T >0$ such that for all $\Rr/(\Rd+\Rno) < \delta$ trajectories $\mu(t)$ of \eqref{eq:ODE_mu_ev} satisfy $d_{\MSNE(F,\phi)}(\mu(t))\leq B\;\;\; \forall t \geq T$. Here, the distance of a point $x$ to a set $\Acal$ is defined as $d_{\Acal}(x):= \inf_{y\in \Acal} \|x-y\|$. Hence, we only need to show that for any $\epsilon >0$ there is  $B>0$  such that $d_{\MSNE(F,\phi)}(\mu)\leq B \implies \mu \in \MSNE_\epsilon(F,\phi)$ for all $\mu \in X$.  Notice that $\MSNE(F,\phi)$ is compact by Theorem~\ref{th:unique_MSNE}, and thus closed. Therefore, the infimum in the definition of distance of a point to $\MSNE(F,\phi)$ is attained in $\MSNE(F,\phi)$. It follows that if $d_{\MSNE(F,\phi)}(\mu)\leq B$ for some $B> 0$ then there is $\nu \in \MSNE(F,\phi)$ such that $\|\mu-\nu\|\leq B$. Thus, henceforth, consider any $\mu\in X$ and any $\nu \in \MSNE(F,\phi)$ such that $\|\mu-\nu\|\leq B$  for a generic $B>0$.
First, for any $c\in [C]$ and any $u\in \UDcbar$, if $\mu^c[\Kcal,u] > 0$ either (i)~$\nu^c[\Kcal,u] > 0$ or (ii)~$\nu^c[\Kcal,u] = 0$. In case~(i), it follows by the definition of MSNE that $F^c_u(\nu)  \geq F^c_v(\nu)$ for all $v\in \UDcbar$. Therefore, since the payoff function is continuous by Assumption~\ref{ass:cont}, it follows that there exists $B_1>0$ such that for all $B\leq B_1$ one has $F^c_u(\mu)  \geq F^c_v(\mu)-\epsilon$. In case~(ii), for all $B\leq B_2 = \epsilon$ one has $\|\mu-\nu\|\leq B_2 = \epsilon$, hence it follows that $\mu^c[\Kcal,u] -\nu^c[\Kcal,u] = \mu^c[\Kcal,u] \leq B_2 = \epsilon$. 
Second,  for all $c\in [C]$, all $u\in \UDcbar$, and all $k\in \Kcal$
\begin{equation*}
	\begin{split}
		& |\mu^c[k,u] - \eta^{c,u}(k)\mu^c[\Kcal,u]| \\
		 = &	|\nu^c[k,u] + (\mu^c[k,u] - \nu^c[k,u]) - \eta^{c,u}(k)\nu^c[\Kcal,u] 
		 \\ & \quad \quad -\eta^{c,u}(k)(\mu^c[\Kcal,u]-\nu^c[\Kcal,u])|\\
		 = &	|(\mu^c[k,u] - \nu^c[k,u]) - \eta^{c,u}(k)(\mu^c[\Kcal,u]-\nu^c[\Kcal,u])|  \leq 2B.
	\end{split}
\end{equation*}
Hence, $|\mu^c[k,u] - \eta^{c,u}(k)\mu^c[\Kcal,u]| \leq \epsilon$  for all $B\leq B_3 = \epsilon/2$. One concludes that by choosing $B = \min(B_1,B_2,B_3)$ the conditions \eqref{eq:epsMSNE_u} and \eqref{eq:epsMSNE_s} of the definition of a $\epsilon$-MSNE hold for $\mu$. \hfill$\qedsymbol$%

\section{Proofs of Section~\ref{sec:design}}

\subsection{Proof of Theorem~\ref{th:opt_intra_class_fair}}\label{sec:proof_th_opt_intra_class_fair}
Consider any class $c\in [C]$ and let $i,j$ represent independent samples drawn from $\Ccal^c$ uniformly at random. Under Assumptions~\ref{ass:cont} and \ref{ass:noise} we are in the conditions of Lemma~\ref{lem:F_def}, so $J^i$ and $J^j$ can be expressed by \eqref{eq:Ji_avg}. For any $\delta \in (0,1)$ let us choose $\epsilon = \delta m^c/|\Ucal_D^c| \leq \delta$. It follows from the definition of $\epsilon$-MSNE in Definition~\ref{def:MSNE} that $J^i(\mu) \geq J^j(\mu) -\delta$ with probability greater than or equal to $1-\epsilon|\Ucal_D^c|/m^c = 1-\delta$ and $J^j(\mu) \geq J^i(\mu) -\delta$ with probability greater than or equal to $1-\epsilon|\Ucal_D^c|/m^c = 1-\delta$ for all $\mu\in \MSNE_\epsilon(F,\phi)$, which immediately establishes that $|J^j(\mu) - J^i(\mu)| \leq \delta$ with probability $(1-\delta)^2$. The same argument can be used to show the statement when $\delta = 0$.\hfill$\qedsymbol$%

\subsection{Proof of Lemma~\ref{lem:opt_prices}}\label{sec:proof_lem_opt_prices} 
The proof of this result resorts to strong duality of linear programming \cite[Chap.~5]{BoydVandenberghe2004}. Consider the dual of \eqref{eq:class_opt} 
\begin{equation}\label{eq:class_opt_dual}
	\begin{aligned}
		\underset{y^c_t, y^c_c}{\min}\quad & {m^c\Rno y^c_t+m^c\Rd y^c_m} \\
		\mathrm{s.t.}\quad  
		& { y_t^c \geq 0}\\
		& { y_m^c\in \R}\\
		& \tilde{\tau}^c_ay^c_t+y^c_m \geq w^\star_a, \quad a \in \Acal^c.
	\end{aligned}
\end{equation}
The idea of the proof is to first show that $f^{c\star}_a$ with $a\in \Acal^c$ satisfies the constraints of \eqref{eq:class_opt} given the tolls in \eqref{eq:opt_prices_cont}. Second, we show that there exist $y^{c\star}_t$ and $y^{c\star}_m$ which satisfy the constraints of the dual \eqref{eq:class_opt_dual} and that are complementary to $f^{c\star}_a$ with $a\in \Acal^c$. Then, by strong LP duality \cite[Chap.~5]{BoydVandenberghe2004} $f^{c\star}_a$ with $a\in \Acal^c$  and $(y^{c\star}_t,y^{c\star}_m)$ are solutions to \eqref{eq:class_opt} and \eqref{eq:class_opt_dual}, respectively.
First, it follows immediately from the constraints of \eqref{eq:sys_opt} that $\sum\nolimits_{a\in \Acal^c} f^{c\star}_a  = m^c\Rd$ and that $f_a^{c\star} \geq 0$ for all $a\in \Acal^c$. Moreover, from \eqref{eq:opt_prices_cont} one can write
\begin{equation}\label{eq:equality_primal}
	\begin{split}
	\!\!\!\sum_{a\in\Acal^c} \!\!f_a^{c\star} \tilde{\tau}_a^c &\!=\! (\Rno/\Rd)\!\! \!\sum_{a\in\Acal^c} \!f_a^{c\star}\!+\! \alpha\!\!\! \sum_{a\in\Acal^c} \!\!f_a^{c\star}  w_a^\star \!-\! \alpha \bar{w}_c^\star\!\!\!\sum_{a\in\Acal^c} \!f_a^{c\star}\!\!\! \\
	& = m^c\Rno \!+\! \alpha m^c\Rd \bar{w}_c^\star \!-\! \alpha  m^c\Rd \bar{w}_c^\star \!=\! m^c\Rno.\!\!\!
\end{split}
\end{equation}
Therefore, one concludes that $f^{c\star}_a$ with $a\in \Acal^c$ satisfies the constraints of \eqref{eq:class_opt} and that the first inequality is active. Second, we show that there is $(y^{c\star}_t,y^{c\star}_m)$ that satisfies the constraints of \eqref{eq:class_opt_dual}.  Consider the candidate solution $y^{c\star}_t = 1/\alpha$ and $y^{c\star}_m = -\Rno/(\alpha \Rd) + \bar{w}_c^\star$. Notice that $y_t^{c\star} \geq 0$ and $y_m^{c\star} \in  \R$. Furthermore, for any $a\in \Acal^c$
\begin{equation}\label{eq:ineq_dual}
	\begin{cases}
		\!\tilde{\tau}^c_ay^c_t \!+\! y^c_m = \frac{\Rno}{\alpha \Rd}\!+\!\frac{\alpha(w_a^\star-\bar{w}_c^\star)}{\alpha} -\frac{\Rno}{\alpha \Rd} \!+\! \bar{w}_c^\star = w_a^\star, & f^{c\star}_a >0\\
		\!\tilde{\tau}^c_ay^c_t\!+\!y^c_m \geq  \frac{\Rno}{\alpha \Rd}+\frac{\alpha(w_a^\star-\bar{w}_c^\star)}{\alpha} -\frac{\Rno}{\alpha \Rd} \!+\! \bar{w}_c^\star \!=\! w_a^\star, & f^{c\star}_a  = 0.
	\end{cases}
\end{equation}
Therefore,  $(y^{c\star}_t,y^{c\star}_m)$ satisfies the constraints of \eqref{eq:class_opt_dual}. Third, we show that $f^{c\star}_a$ with $a\in \Acal^c$ and $(y^{c\star}_t,y^{c\star}_m)$ are complementary. Indeed $f^{c\star}_a >0 \implies \tilde{\tau}^c_ay^c_t+y^c_m = w_a^\star$ for all $a\in \Acal^c$ by \eqref{eq:ineq_dual}. Conversely, $\tilde{\tau}^c_ay^c_t+y^c_m > w_a^\star \implies f^{c\star}_a = 0 $ for all $a\in \Acal^c$ also by \eqref{eq:ineq_dual}. Moreover, $y_t^c = 1/\alpha >0$ and $\sum_{a\in\Acal^c} f_a^{c\star} \tilde{\tau}_a^c = m^c\Rno$ by \eqref{eq:equality_primal}. Then, by strong LP duality \cite[Chap.~5]{BoydVandenberghe2004}, $f^{c\star}_a$ with $a\in \Acal^c$  is a solution to \eqref{eq:class_opt}.\hfill$\qedsymbol$%


\subsection{Proof of Lemma~\ref{lem:opt_pol}}\label{sec:proof_lem_opt_pol}

This proof relies on three propositions, whose proofs are presented at the end of the section. The proof of Lemma~\ref{lem:opt_pol} is split in two cases, depending on whether the following condition is satisfied for each class.
\begin{condition}\label{cond:distinct_payoff}
	For a given class $c\in [C]$ there is $a \in \Acal^c$  such that $f^{c\star}_{a}>0$ and $w_a^\star \neq  \bar{w}_c^\star$.
\end{condition}%
First, under Condition~\ref{cond:distinct_payoff}, Proposition~\ref{prop:decomp} shows that after discretizing the tolls of Lemma~\ref{lem:opt_prices}, it is possible to decompose the optimal flows of each class $c\in[C]$ with arbitrary precision into a finite number of flows $g^{c,i}_a$ with $i = 1,2,\ldots,L$ and $a\in \Acal^c$ such that for each $i$: (i)~two and only two actions are used; (ii)~the rate of earning tokens associated with the flows is null; (iii)~the average payoff of $g^{c,i}_a$ with $a\in \Acal^c$ can be made arbitrarily close to the average payoff of $f^{c~\star}_a$  with $a\in \Acal^c$. 


\begin{proposition}\label{prop:decomp}
	Under Assumption~\ref{ass:noise}, for $c\in [C]$ under Condition~\ref{cond:distinct_payoff}, consider a discretization $\tau^c_a(\alpha) = \mathrm{round}(\tilde{\tau}^c_a(\alpha))$ of  the continuous tolls in \eqref{eq:opt_prices_cont} in Lemma~\ref{lem:opt_prices} with scaling factor $\alpha>0$ for all $a\in \Acal^c$. Then for all $\epsilon>0$ there is $\alpha^\star>0$ such that for all $\alpha\geq\alpha^\star$ there exists a finite number $L$ of action flows $g^{c,i}_a \geq 0$ with $i = 1,2,\ldots,L$ and $a\in \Acal^c$ such that 
	\begin{enumerate}[(a)]
		\item For all $i = 1,2,\ldots, L$, there exist two and only two nonnull $g^{c,i}_a$ with $a\in \Acal^c$;
		\item For all $i = 1,2,\ldots, L$, $m^c\Rno {\sum_{a\in \Acal^c} g^{c,i}_a}/({m^c\Rd}) -  \sum_{a\in \Acal^c} \tau^c_a(\alpha) g^{c,i}_a = 0$ 
		\item $\sum_{i = 1}^L \sum_{a\in \Acal^c}g^{c,i}_a = m^c\Rd$;
		\item For all $a\in \Acal^c$, $\big|\sum_{i = 1}^L g_a^{c,i} - f_a^{c\star}  \big |\leq \epsilon$;
		\item For all $i \!=\! 1, 2, \ldots, L$, $\big|\!\sum_{a\in \Acal}\!w_a^\star g_a^{c,i} \!- \bar{w}^\star \sum_{a\in \Acal} g_a^{c,i}     \big | \!\leq \epsilon$.%
	\end{enumerate}
\end{proposition}

Second, we consider each of the pairs of flows that result from the decomposition detailed in Proposition~\ref{prop:decomp}. We show that for a sufficiently large maximum amount of tokens $\bar{k}$, there is a policy whose stationary action distribution can be made arbitrarily close to the proportion of each pair of flows of the decomposition.

\begin{proposition}\label{prop:there_is_a_pol}
	Under Assumptions~\ref{ass:senseless_policies} and~\ref{ass:noise}, for $c\in [C]$ under Condition~\ref{cond:distinct_payoff} consider a discretization $\tau^c_a = \mathrm{round}(\tilde{\tau}^c_a)$ of the continuous tolls in \eqref{eq:opt_prices_cont} in Lemma~\ref{lem:opt_prices} for some scaling factor $\alpha>0$ for all $a\in \Acal_c$. Consider a pair of actions $p,q \in \Acal^c$ and flows $g_p^c,g_q^c>0$ with unitary sum, i.e., $g_p^c+ g_q^c = 1$, such that $\Rno/\Rd - \tau^c_p g^{c}_p -  \tau^c_q g^{c}_q = 0$. Then, for all $\epsilon >0$ there exists $\bar{k}^\star \in \N$ such that for all $\bar{k}\geq \bar{k}^\star$ there is a policy $u\in \UDcbar$ such that 
	\begin{equation*}
		\Big | g_a^c-  \sum_{k\in \Kcal} \eta^{c,u}(k)u(a|k)\Big| \leq \epsilon, \quad \forall a\in \{p,q\}.
	\end{equation*}
\end{proposition}

Third, Proposition~\ref{prop:no_pol_is_better} shows that, for each class $c\in[C]$, after discretizing the tolls of Lemma~\ref{lem:opt_prices} no policy exists that can achieve better payoff that the optimal value of \eqref{eq:class_opt} up to an arbitrary precision $\epsilon >0$.

\begin{proposition}\label{prop:no_pol_is_better}
	Under Assumptions~\ref{ass:senseless_policies} and~\ref{ass:noise}, for any $c\in [C]$ consider a discretization $\tau^c_a(\alpha) = \mathrm{round}(\tilde{\tau}^c_a(\alpha))$ of  the continuous tolls in \eqref{eq:opt_prices_cont} in Lemma~\ref{lem:opt_prices} with scaling factor $\alpha>0$ for all $a\in \Acal_c$. Then for all $\epsilon>0$ there is $\alpha^\star>0$ such that for all $\alpha\geq\alpha^\star$
	\begin{equation*}
		\bar{w}^\star  \geq \sum_{k\in \Kcal} \sum_{a\in \Acal^c} \eta^{c,v}(k)v(a|k)w^\star_a -\epsilon,  \quad  \forall  v\in \UDcbar.
	\end{equation*}
\end{proposition}

Finally, we are ready to show the statement of the lemma. Consider a class $c\in [C]$. 
First, assume that Condition~\ref{cond:distinct_payoff} is not satisfied for class $c$. Then, for each $a\in \Acal^c$ such that $f^{c^\star}_a >0$, $\tau^c_a = 0$ and one can allocate in $\mu^c \in X^c$ the mass $f^{c^\star}_a$ to the policy that uses action $a$ from any token amount. Then, \eqref{eq:opt_pol_cond1} is satisfied by Proposition~\ref{prop:no_pol_is_better} under Assumptions~\ref{ass:senseless_policies} and~\ref{ass:noise}. Moreover, the left-hand side of \eqref{eq:opt_pol_cond2} is zero by construction, therefore \eqref{eq:opt_pol_cond2} is satisfied for any $\epsilon >0$.
Second, assume that Condition~\ref{cond:distinct_payoff} is satisfied for class $c$. Then, under Assumptions~\ref{ass:senseless_policies} and~\ref{ass:noise} by Propositions~\ref{prop:decomp}, \ref{prop:there_is_a_pol}, and~\ref{prop:no_pol_is_better} for sufficiently large $\alpha$ and sufficiently large $\bar{k}$ there is a finite number of policies that one can allocate in $\mu^c \in X^c$ such that  \eqref{eq:opt_pol_cond1} and \eqref{eq:opt_pol_cond2} are satisfied. To conclude, since there is a finite number of classes, the result holds uniformly by taking the largest thresholds $\alpha^\star$ and $\bar{k}^\star$ among all classes. \hfill  $\qedsymbol$%

	\subsubsection{Proof of Proposition~\ref{prop:decomp}}
		Under the conditions of the proposition, from Lemma~\ref{lem:opt_prices},  $f^{c\star}_a$ with $a\in \Acal^c$ is a solution to \eqref{eq:class_opt}. First, we show that, under Condition~\ref{cond:distinct_payoff}, the inequality $m^c\Rno- \sum\nolimits_{a\in \Acal^c} \tilde{\tau}_a^c(\alpha) f^{c\star}_a \geq 0$ in \eqref{eq:class_opt} is nonstrict at the solution for class $c$. We prove this claim by contradiction. Assume, by contradiction, that the inequality is strict, i.e., $m^c\Rno- \sum\nolimits_{a\in \Acal^c} \tilde{\tau}_a^c(\alpha) f^{c\star}_a > 0$, and that Condition~\ref{cond:distinct_payoff} holds. Then, from Condition~\ref{cond:distinct_payoff}, there exist $p,q \in \Acal^c$ such that $f^{c\star}_p>0$, $f^{c\star}_q>0$, and $w_p^\star >  w_q^\star$. Consider a perturbation of the solution to \eqref{eq:class_opt} characterized by $\nu^c_a = f^{c\star}_a$ for all $a\notin  \{p,q\}$,  $\nu^c_p = f^{c\star}_p+\delta$, and $\nu^c_q = f^{c\star}_q-\delta$, where $\delta>0$ is small enough such that $\nu^c_q >0$. Notice that $\nu^c_a$ with $a\in \Acal^c$ is in the feasible set of \eqref{eq:class_opt} and that it leads to a strictly lower objective than the one achieved by  $f^{c\star}_a$ with $a\in \Acal^c$, which is a contradiction with the hypothesis that  $f^{c\star}_a$ with $a\in \Acal^c$ is a solution to \eqref{eq:class_opt}.
		
		Second, we show that there exists a decomposition of flows $\tilde{g}_a^{c,i}$ with $i = 1,2,\ldots,L$ that satisfies (b)--(e) with $\epsilon =0$  if $\tilde{\tau}^c_a(\alpha) = \tau^c_a(\alpha)$. One can show this claim with a constructive argument. Indeed,  one can sequentially allocate flow from $f^{c\star}$  to: (i)~a single action $p$ whose toll is $\tilde{\tau}^c_p(\alpha) = \Rno/\Rd$ or (ii)~pairs of actions $(p,q)$ whose tolls satisfy $\tilde{\tau}^c_p(\alpha) < \Rno/\Rd$ and $\tilde{\tau}^c_q(\alpha) > \Rno/\Rd$. For a single action $p\in \Acal^c$ in the $i$-th allocation, one can allocate the entirety of the flow in $f^{c\star}_p$, i.e., $\tilde{g}^{c,i}_p = f^{c\star}_p$ and  $\tilde{g}^{c,i}_a =0$ for all $a\neq p$. Notice that (b) is satisfied for $i$ by construction if $\tilde{\tau}^c_p(\alpha) = \tau^c_p(\alpha)$. For two actions $p_i,q_i \in \Acal^c$ in the $i$-th allocation, one can set $\tilde{g}^{c,i}_a =0$ for all $a\notin\{ p_i,q_i\}$ and choose $\tilde{g}^{c,i}_{p_i},\tilde{g}^{c,i}_{q_i} >0$ such that
		\begin{equation*}
			m^c\Rno (\tilde{g}^{c,i}_{p_i}+\tilde{g}^{c,i}_{q_i}) /(m^c\Rd) - \tilde{g}^{c,i}_{p_i}\tilde{\tau}^c_{p_i}(\alpha) - \tilde{g}^{c,i}_{q_i}  \tilde{\tau}^c_{q_i}(\alpha) = 0
		\end{equation*}
		and that all the remaining unallocated flow of either $p_i$ or $q_i$ is entirely allocated in $i$. Again notice~(a) and~(b) are satisfied for $i$ by construction if  $\tilde{\tau}^c_{p_i}(\alpha) = \tau^c_{p_i}(\alpha)$ and $\tilde{\tau}^c_{q_i}(\alpha) = \tau^c_{q_i}(\alpha)$. Since  $m^c\Rno- \sum\nolimits_{a\in \Acal^c} \tilde{\tau}_a^c f^{c\star}_a = 0$ it is always possible to find a new action or pair of actions until all the flow in  $f^{c\star}_a$ with $a\in \Acal^c$ is allocated, i.e., $\sum_{i = 1}^L \tilde{g}_a^{c,i} = f_a^{c\star}$ for all $a\in \Acal^c$ thereby satisfying~(c) and~(d) with $\epsilon =0$. Furthermore, notice that for all $i = 1,2,\ldots,L$ we have shown that
		\begin{equation*}
			m^c\Rno {\sum\nolimits_{a\in \Acal^c} \tilde{g}^{c,i}_a}/({m^c\Rd}) -  \sum\nolimits_{a\in \Acal^c} \tilde{\tau}^c_a(\alpha) \tilde{g}^{c,i}_a = 0, 
		\end{equation*}
		which can be rewritten using \eqref{eq:opt_prices_cont} as
		\begin{equation*}
			\begin{split}
				&\frac{\Rno}{\Rd} \!\sum_{a\in \Acal^c}\! \tilde{g}^{c,i}_a \!-\! \frac{\Rno}{\Rd}\!\!\sum_{a\in \Acal^c} \tilde{g}^{c,i}_a \!-\! \alpha\! \sum_{a\in \Acal^c} \!\tilde{g}^{c,i}_a \left(w^{\star}_a \!-\! \bar{w}_c^\star \right) = 0\\
				\!\iff &
				\sum\nolimits_{a\in \Acal^c} w^{\star}_a  \tilde{g}^{c,i}_a  -  \bar{w}_c^\star \sum\nolimits_{a\in \Acal^c}  \tilde{g}^{c,i}_a   = 0.
			\end{split}
		\end{equation*}
		Therefore,~(e) is also satisfied with $\epsilon = 0$.
		
		Third, we write a candidate flow decomposition $g_a^{c,i}$ with $i = 1,2,\ldots,L$ and $a\in \Acal^c$ for the discretized tolls as a solution to the following quadratic program (QP)
		\begin{equation}\label{eq:decomp_QP}
			\begin{aligned}
				\!\!\underset{g_a^{c,i}}{\min}\quad \!\!& \sum_{a\in \Acal^c}\!\! \left(\sum_{i = 1}^L g_a^{c,i} \!-\! f_a^{c\star}\!\right)^{\!\!2} \!\!+ \!\sum_{i = 1}^L \!\left(\sum_{a\in \Acal} \!w_a^\star g_a^{c,i} \!-\! \bar{w}^\star\! \sum_{a\in \Acal}\! g_a^{c,i} \!\right)^{\!\!2} \!\!\!\!\!\!\!\!\!\!\\
				\!\!\mathrm{s.t.}  \quad 
				& \frac{\Rno}{\Rd} {\sum_{a\in \Acal^c} \!g^{c,i}_a} \!-\!  \sum_{a\in \Acal^c}\! \tau^c_a(\alpha) g^{c,i}_a \!= \!0, \; i = 1,2,\ldots,L\\
				&  \sum\nolimits_{i = 1}^L \sum\nolimits_{a\in \Acal^c}g^{c,i}_a = m^c\Rd \\
				& g^{c,i}_a \geq 0, \quad  i = 1,2,\ldots,L,  \;  a\in \Acal^c,
			\end{aligned}
		\end{equation} %
		where $g_a^{c,i}$ for $i = 1,2,\ldots,L$ are all combinations where at most two entries $g_a^{c,i}$ with $a\in \Acal^c$ are nonnull. Under Condition~\ref{cond:distinct_payoff}, there is always $\tilde{\tau}_p^c(\alpha) > \Rno/\Rd$ and $\tilde{\tau}_q^c(\alpha)< \Rno/\Rd$ with $p,q\in \Acal^c$. After the discretization, for sufficiently large $\alpha$, i.e., $\alpha \geq \alpha^\star_1$, there is also always ${\tau}_p^c(\alpha) > \Rno/\Rd$ and ${\tau}_q^c(\alpha)< \Rno/\Rd$. It follows that the feasible set of \eqref{eq:decomp_QP}  is nonempty for all $\alpha\geq \alpha^\star_1$ and compact, since it is characterized by linear equalities and nonstrict linear inequalities. Since the QP \eqref{eq:decomp_QP} has a continuous objective function and nonempty compact feasible set, then a solution exists. From the constraints of \eqref{eq:decomp_QP}, it follows immediately that a solution satisfies~(b)--(c). To show that is also satisfies (a), assume by contradiction that there is a solution to \eqref{eq:decomp_QP} for which there is $g^{c,i}_a$ with $a\in \Acal^c$ with a single nonnull entry $g^{c,i}_p >0$. Then, one can rewrite the first constraint of \eqref{eq:decomp_QP} as $(\Rno/\Rd- \tau_p^c(\alpha)) g^{c,i}_p = 0$, which is equivalent to $\Rno/\Rd =  \tau_p^c(\alpha)$. But, by Assumption~\ref{ass:noise}, $0< \Rno/\Rd <1$, which is a contradiction with the hypothesis that $\tau_p^c(\alpha)$ is an integer. To show that it also satisfies~(d) and~(e) for sufficiently large $\alpha$, we can rewrite the first equality constraint of \eqref{eq:decomp_QP} using the identity $\tau^c_a(\alpha) = \tilde{\tau}^c_a(\alpha) + \Delta_a$ for all $a\in \Acal^c$, where $\Delta_a$ is the discretization error that satisfies $|\Delta_a|\leq1/2$. We can then write for each $i = 1,2,\ldots, L$
		\begin{equation*}
			\begin{split}
				\frac{\Rno}{\Rd} {\sum_{a\in \Acal^c} g^{c,i}_a}-  \sum_{a\in \Acal^c} (\tilde{\tau}^c_a(\alpha) + \Delta_a) g^{c,i}_a & \!= \!0 \\
				\frac{\Rno}{\Rd}\! {\sum_{a\in \Acal^c} \!\!g^{c,i}_a}  \!-\!  	\frac{\Rno}{\Rd} \!{\sum_{a\in \Acal^c}\!\! g^{c,i}_a}\!-\! \alpha\!\! \!\sum_{a\in \Acal^c} \!\! (w_a^\star\!-\!\bar{w}_c^\star)g^{c,i}_a \!-\!\!\! \sum_{a\in \Acal^c}\! \!\Delta_a g^{c,i}_a  & \!=\! 0 \\
				\sum_{a\in \Acal^c}  (w_a^\star-\bar{w}_c^\star)g^{c,i}_a + \sum_{a\in \Acal^c} \theta_a g^{c,i}_a  &\! =\! 0,
			\end{split}
		\end{equation*}
		where we use the change of variables $\theta_a := \Delta_a/\alpha$ for all $a\in \Acal^c$. As a result, one can equivalently write \eqref{eq:decomp_QP} as 
		\begin{equation}\label{eq:decomp_QP_param}
			\begin{aligned}
				\!\!\!\underset{g_a^{c,i}}{\min}\; & \sum_{a\in \Acal^c}\!\!\left(\sum_{i = 1}^L \!g_a^{c,i} \!-\! f_a^{c\star}\!\!\right)^{\!2} \!\!\!+\! \sum_{i = 1}^L\;\left(\sum_{a\in \Acal} \!w_a^\star g_a^{c,i} \!-\! \bar{w}^\star \! \sum_{a\in \Acal} \!g_a^{c,i} \!\right)^{\!2}\!\!\!\!\!\!\!\!\!\!\! \\
				\mathrm{s.t.}\;\,   
				&  \sum_{a\in \Acal^c}  (w_a^\star +\theta_a-\bar{w}_c^\star)g^{c,i}_a = 0, \; i = 1,2,\ldots,L\!\!\!\!\\
				&  \sum\nolimits_{i = 1}^L \sum\nolimits_{a\in \Acal^c}g^{c,i}_a = m^c\Rd \\
				& g^{c,i}_a \geq 0, \quad  i = 1,2,\ldots,L,  \;  a\in \Acal^c.
			\end{aligned}
		\end{equation}
		One can conclude that when $\theta_a = 0$ for all $a\in \Acal^c$, which corresponds to null discretization error, then the decomposition of flows $\tilde{g}_a^{c,i}$ with $i = 1,2,\ldots,L$ and $a\in \Acal^c$ is a solution to \eqref{eq:decomp_QP_param} and it achieves null objective at the solution.

		Fourth, we study how the optimal value of the objective of \eqref{eq:decomp_QP_param} changes w.r.t.\ the parameters $\theta_a$ with $a\in \Acal^c$. The objective of \eqref{eq:decomp_QP_param} is continuous w.r.t.\ the decision variables and does not depend on the parameters $\theta_a$ with $a\in \Acal^c$, the feasible set is nonempty and compact for $\alpha\geq \alpha_1^\star$ (i.e., for $\theta_a$ with $a\in \Acal^c$ sufficiently close to zero), the feasible set is continuous (upper and lower hemicontinuous) at $\theta_a = 0$  for all $a\in \Acal^c$ since Condition~\ref{cond:distinct_payoff} holds. We are in the conditions of Berges's Maximum Theorem \cite[Theorem~17.31]{AliprantisBorder2006}, so the optimal objective is continuous w.r.t.\ the parameters $\theta_a$ with $a\in \Acal^c$ at $\theta_a = 0$  for all $a\in \Acal^c$. Finally, by increasing $\alpha$,   $\theta_a$  for all $a\in \Acal^c$ can be made arbitrarily close to zero and, by continuity, the solution to \eqref{eq:decomp_QP_param} has an objective that can be made arbitrarily close to zero. It follows that for all $\epsilon>0$ that there exists $\alpha^\star$ such that for all $\alpha\geq \alpha^\star$, a solution to \eqref{eq:decomp_QP_param} exists and its objective is at most $\epsilon$. It follows immediately that such solution satisfies~(a)--(e).


	\subsubsection{Proof of Proposition~\ref{prop:there_is_a_pol}}
		First, we claim that $\tau^c_q \neq \tau^c_p$ and that either $\tau^c_q$ or $\tau^c_q$ is nonpositive. To see why, assume by contradiction that (i)~$\tau^c_q = \tau^c_p$; or (ii) $\tau^c_q$ and $\tau^c_q$ are both positive. Consider that~(i) holds. By hypothesis, $0 =  \Rno/\Rd - \tau^c_p g^{c}_p -  \tau^c_q g^{c}_q = \Rno/\Rd  - \tau^c_p$, or equivalently, $\Rno/\Rd = \tau^c_p$. By Assumption~\ref{ass:noise}, it follows that $0<\Rno/\Rd <1$, which is a contradiction with the fact that $\tau^c_p$ is an integer. Consider that~(ii) holds, i.e,  $\tau^c_q \geq 1$ and  $\tau^c_p \geq 1$. Therefore, $0 =  \Rno/\Rd - \tau^c_p g^{c}_p -  \tau^c_q g^{c}_q \leq \Rno/\Rd  - g^{c}_p - g^{c}_q  = \Rno/\Rd-1$, or equivalently that $\Rno/\Rd \geq 1$, which is a contradiction with Assumption~\ref{ass:noise}.  Henceforth, consider without any loss of generality that $\tau^c_p\leq 0$ and  $\tau^c_q >0$. 
		
		%
		
		Second, the proof of this result is constructive. Consider a candidate policy $u\in \UDcbar$ defined as $u(p|k) = 1$ if $0\leq k < \tau^c_q$ and $u(q|k) = 1$ if $\tau^c_q \leq k \leq \bar{k}$. The stationary state distribution of $u$, denoted by $\eta^{c,u} \in \Pcal(\Kcal)$, exists and is unique by Lemma~\ref{lem:F_def} under Assumption~\ref{ass:noise}. Moreover, the stationary state distribution follows the token stationarity condition
		\begin{equation*}
			\begin{split}
					&\sum_{k = 0}^{\tau_q^c-1}\! \eta^{c,u}(k)\!\left(\!-\frac{\Rd}{\Rno\!+\!\Rd}\tau^c_p\!+\! \frac{\Rno}{\Rno\!+\!\Rd}\right) \!-\!  \eta^{c,u}(\bar{k})\frac{\Rd}{\Rno\!+\!\Rd} \tau^c_q\\
					&+ \sum_{k = \tau_q^c}^{\bar{k}-1} \!\eta^{c,u}(k)\left(-\frac{\Rd}{\Rno\!+\!\Rd}\tau^c_q+ \frac{\Rno}{\Rno\!+\!\Rd}\right)  = 0,
			\end{split}
		\end{equation*}
		which can be equivalently rewritten as
		\begin{equation}\label{eq:two_action_pol_ss}
			-\tilde{g}_p^c  \tau_p^c -\tilde{g}_q^c  \tau_q^c + (1- \eta^{c,u}(\bar{k})) \Rno/\Rd  = 0,
		\end{equation}
		where $\tilde{g}^c_p:= \sum_{k = 0}^{\tau_q^c-1} \eta^{c,u}(k)$ and $\tilde{g}^c_q:= \sum_{k = \tau_q}^{\bar{k}} \eta^{c,u}(k)$ are the flows of actions $p$ and $q$, respectively. Since $\tilde{g}^c_p + \tilde{g}^c_q = 1$, \eqref{eq:two_action_pol_ss} uniquely characterizes the action distribution of the candidate policy. Let us consider a solution parameterized by a scalar $\Delta \in \R$ as $\tilde{g}^c_p = g^c_p +\Delta$ and  $\tilde{g}^c_q = g_q^c -\Delta$. One can then rewrite \eqref{eq:two_action_pol_ss} as
		\begin{equation}\label{eq:two_action_pol_ss_delta}
			\begin{split}
				& \frac{\Rno}{\Rd} -g^c_p  \tau_p^c -g^c_q  \tau_q^c  - \Delta \tau_p^c  - \Delta \tau_q^c \eta^{c,u}(\bar{k})\frac{\Rno}{\Rd} \!=\! 0\!\\
				\iff & \Delta \tau_p^c  - \Delta \tau_q^c\eta^{c,u}(\bar{k})\Rno/\Rd  = 0\\
				\iff & \Delta = \frac{\Rno}{\Rd} \frac{1}{ \tau_q^c- \tau_p^c}\eta^{c,u}(\bar{k}),
			\end{split}
		\end{equation}
		which is well defined since it was already established that $\tau^c_q \neq \tau^c_p$. As a result, the action distribution of the candidate policy $u$ differs from the flows $g_p^c,g_q^c>0$ by $\Delta$, which by \eqref{eq:two_action_pol_ss_delta} is proportional to the mass on the maximum amount of tokens.
		
		Third, we show that $\eta^{c,u}(\bar{k})$ can be made arbitrarily small as $\bar{k}$ increases. To do so, we split the token support in $\{\tau^c_q,\tau^c_q+1,\ldots, \bar{k}\}$ into chunks of mass. Specifically define $\eta_s := \sum_{k = \bar{k}-s\tau^c_q+1}^{\bar{k}-(s-1)\tau^c_q}\eta^{c,u}(k)$ for $s = 1,2,\ldots, \lfloor (\bar{k}-\tau_q^c+1)/\tau^c_q\rfloor$. We proceed to prove by induction that $\eta_{s} \leq \eta_{s+1}\gamma$, where $\gamma := 1/2-\sqrt{1/4 -\Rno/\Rd}$. Notice that $\gamma < 1/4 +\Rno/\Rd < 1$ because by Assumption~\ref{ass:noise} $\sqrt{1/4 -\Rno/\Rd} > 1/4 -\Rno/\Rd$. Moreover,
		$\gamma > \Rno/\Rd$ because by Assumption~\ref{ass:noise}
		\begin{equation*}
			\begin{split}
				&0 < (\Rno/\Rd)^2 \\
				\!\iff&  1/4\!-\!\Rno/\Rd < \!1/4\!-\!\Rno/\Rd \!+\!(\Rno/\Rd)^2\\ 
				\!\iff &  \sqrt{1/4\!-\!\Rno/\Rd}  < \!1/2\!-\!\Rno/\Rd\\
				\!\iff &\gamma > \Rno/\Rd.
			\end{split}
		\end{equation*}
		Finally, $\gamma = (\Rno/\Rd)/(1-\gamma)$ because
		\begin{equation*}
			\begin{split}
				&\frac{\Rno}{\Rd}\frac{1}{1-\gamma} = 	\frac{\Rno}{\Rd}\frac{1}{1/2+\sqrt{1/4-\Rno/\Rd}} \\ =& \frac{\Rno}{\Rd}\frac{1/2-\sqrt{1/4-\Rno/\Rd}}{1/4 - 1/4 +\Rno/\Rd} = \gamma
			\end{split}
		\end{equation*}
		For $s = 1$, the stationary distribution of the last chunk satisfies
		\begin{equation*}
			-\eta_1 \frac{\Rd}{\Rd+\Rno}+ \eta^{c,u}(\bar{k}-\tau^c_q)\frac{\Rno}{\Rd+\Rno} = 0. 
		\end{equation*}
		Since $ \eta^{c,u}(\bar{k}-\tau^c_q) \leq \eta_2$, it follows that $\eta_{1} \leq \eta_{2}\Rno/\Rd$, which implies that $\eta_{1} \leq \eta_{2}\gamma$ since  $\gamma > \Rno/\Rd$. The stationary distribution of the $s$-th chunk satisfies
		\begin{equation}\label{eq:ss_s_chunck}
			\begin{split}
							&-\eta_s \frac{\Rd}{\Rd+\Rno} -\eta^{c,u}(\bar{k}-(s-1)\tau^c_q) \frac{\Rno}{\Rd+\Rno}  \\
							&\quad\quad + \eta_{s-1} \frac{\Rd}{\Rd+\Rno}   + \eta^{c,u}(\bar{k}-s\tau^c_q)\frac{\Rno}{\Rd+\Rno}  = 0. 
			\end{split}
		\end{equation}
		Using the induction hypothesis, i.e., $\eta_{s-1}\leq \gamma \eta_s$, and since $\eta^{c,u}(\bar{k}-s\tau^c_q) \leq \eta_{s+1}$, one can equivalently rewrite \eqref{eq:ss_s_chunck} as
		\begin{equation*}
			\begin{split}
				\eta_s\Rd &\leq \gamma \Rd \eta_{s} + \Rno\eta_{s+1}\\
				\iff 	\eta_s &\leq  \eta_{s+1} (\Rno/\Rd)/(1-\gamma) = \gamma\eta_{s+1},
			\end{split}
		\end{equation*}
		which concludes the proof by induction. One concludes that $	1 \geq  \eta_{\lfloor (\bar{k}-\tau_q^c+1)/\tau^c_q\rfloor} \geq \gamma^{\lfloor (\bar{k}-\tau_q^c+1)/\tau^c_q\rfloor-1}\eta_1\geq \gamma^{\lfloor (\bar{k}-\tau_q^c+1)/\tau^c_q\rfloor-1} \eta^{c,u}(\bar{k})$,
		i.e., $\eta^{c,u}(\bar{k}) \leq  \gamma^{\lfloor (\bar{k}-\tau_q^c+1)/\tau^c_q\rfloor-1}$. It follows that for all $\epsilon >0$ there exists $\bar{k}^\star$ such that  $\Delta \leq \epsilon$ for all $\bar{k} \geq \bar{k}^\star$, which concludes the proof.

	\subsubsection{Proof of Proposition~\ref{prop:no_pol_is_better}}
	In case Condition~\ref{cond:distinct_payoff} is not satisfied for class $c$, then $w_a^\star = \bar{w}_c^\star$ for all $a\in \Acal^c$ for which $f_a^{c\star} > 0$. First, we show that $w_a \leq \bar{w}^\star$ for all $a\in \Acal^c$ by contradiction. Admit, by contradiction, that there is $p\in \Acal^c$ such that $w_p^\star > \bar{w}_c^\star$. Notice that $f^{c\star}_p = 0$, otherwise one reaches a contradiction immediately since Condition~\ref{cond:distinct_payoff} does not hold by hypothesis. Consider action flows $\nu_p^c = \delta$, where $\delta >0$ can be made arbitrarily small, $\nu_q^c = f_q^{c\star} - \delta$, where $q \in \Acal^c$ is such that $f^{c^\star}_q>0$, and $\nu_a^c = f_a^{c\star}$ for all $a\in \Acal^c\setminus \{p,q\}$. Notice that $\sum\nolimits_{a\in \Acal^c} \nu^c_a  = m^c\Rd$ and that $\nu_a^c \geq 0$ for all $a\in \Acal^c$ if one chooses $\delta$ sufficiently small. Furthermore, by Assumption~\ref{ass:noise} it follows that ${\tau}^c_a = 0$ for all $a\in \Acal^c$ for which $f_a^{c\star} > 0$, thus $m^c\Rno- \sum\nolimits_{a\in \Acal^c} {\tau}_a^c f^{c\star}_a  =  m^c\Rno >0$. As a result, $m^c\Rno- \sum\nolimits_{a\in \Acal^c} {\tau}_a^c \nu^{c}_a  >0$ for sufficiently small $\delta$. One concludes that  $\nu_a^c$ for $a\in \Acal^c$ satisfies the constraints of \eqref{eq:class_opt} and achieves a strictly greater objective than the solution $f_a^{c\star}$  for $a\in \Acal^c$, which is a contradiction. Second, given that $w_a \leq \bar{w}_c^\star$ for all $a\in \Acal^c$, it follows immediately that no policy can achieve a payoff that is greater than $\bar{w}_c^\star$, which proves the result for any choice of $\alpha >0$.
	
	In case Condition~\ref{cond:distinct_payoff} is satisfied for class $c$, we rewrite the optimization problem \eqref{eq:class_opt} for the discretized tolls as 
	\begin{equation}\label{eq:class_opt_disc}
		\begin{aligned}
			\underset{\nu^c_a \geq 0, a\in \Acal^c}{\max}\quad & { \sum\nolimits_{a\in \Acal^c} w_a^\star \nu^c_a} \\
			\mathrm{s.t.}\quad  \quad 
			& {  m^c\Rno- \sum\nolimits_{a\in \Acal^c} \tau_a^c(\alpha) \nu^c_a \geq 0}\\
			& {  \sum\nolimits_{a\in \Acal^c} \nu^c_a  = m^c\Rd}
		\end{aligned}
	\end{equation}
	and denote a solution by $\nu^{c\star}_a$ with $a\in \Acal^c$. First, we show that 
	\begin{equation}\label{eq:pol_payoff_ineq1}
		\sum_{a\in \Acal^c} \frac{\nu^{c\star}_{a}w^{\star}_{a}}{m^c\Rd}\geq \sum_{k\in \Kcal} \sum_{a\in \Acal^c} \eta^{c,v}(k)v(a|k)w^\star_a\; \quad \forall v\in \UDcbar
	\end{equation}
	by contradiction, where $\eta^{c,v}$ is well defined by Lemma~\ref{lem:F_def} under Assumption~\ref{ass:noise}. Specifically, assume by contradiction that there is $v\in \UDcbar$ such that $ 	\sum_{a\in \Acal^c} \nu^{c\star}_{a}w^{\star}_{a}/(m^c\Rd) < \sum_{k\in \Kcal} \sum_{a\in \Acal^c} \eta^{c,v}(k)v(a|k)w^\star_a$. Define $\nu_a^c := m^c\Rd\sum_{k\in \Kcal} v(a|k)\eta^{c,v}(k)$ for $a\in \Acal^c$. Notice that $\nu_a^c \geq 0$ for all $a\in \Acal^c$ and that  $\sum\nolimits_{a\in \Acal^c} \nu^c_a  = m^c\Rd$. Moreover, the stationary token distribution of policy $v$ satisfies the flow balance condition 
	\begin{equation*}
		\begin{split}
			&\!\sum_{a\in \Acal^c}\sum_{k\in \Kcal} \!m^c\Rd\eta^{c,v}(k)v(a|k)\tau^c_a(\alpha) +\!\!\!\!\! \sum_{k\in \Kcal\setminus\{\bar{k}\}} \!\!\!\!m^c\Rno\eta^{c,v}(k) =0\!\! \\
			& \iff \sum_{a\in \Acal^c} \nu^c_a\tau^c_a(\alpha) + m^c\Rno = m^c\Rno\eta^{c,v}(\bar{k}) \geq 0,
		\end{split}
	\end{equation*}
	therefore $m^c\Rno- \sum\nolimits_{a\in \Acal^c} \tau_a^c(\alpha) \nu^c_a \geq 0$. One concludes that $\nu_a^c$ for $a\in \Acal^c$ satisfies the constraints of \eqref{eq:class_opt_disc} and the hypothesis   $\sum\nolimits_{a\in \Acal^c} w_a^\star \nu^c_a < m^c\Rd\sum_{k\in \Kcal} \sum_{a\in \Acal^c} \eta^{c,v}(k)v(a|k)w^\star_a =  \sum_{a\in \Acal^c}\nu^{c}_aw^\star_a$, i.e., the objective function of \eqref{eq:class_opt_disc} for $\nu^{c}_a$ with $a\in \Acal^c$ is strictly greater that at the solution  $\nu^{c\star}_a$ with $a\in \Acal^c$, which is a contradiction. Second, we can rewrite the first constraint of \eqref{eq:class_opt_disc} using the identity $\tau^c_a(\alpha) = \tilde{\tau}^c_a(\alpha) + \Delta_a$ for all $a\in \Acal^c$, where $\Delta_a$ is the discretization error that satisfies $|\Delta_a|\leq1/2$, as $m^c\Rno-\sum_{a\in \Acal^c} (\Rno/\Rd +\alpha (w^{\star}_a -\bar{w}^\star) + \Delta_a)\nu^c_a  \geq 0
				\iff  \sum_{a\in \Acal^c} (w^{\star}_a + \theta_a -\bar{w}^\star)  \leq  0$,
	where we use the change of variables $\theta_a := \Delta_a/\alpha$ for all $a\in \Acal^c$. As a result, one can equivalently write \eqref{eq:class_opt_disc} as 
	\begin{equation}\label{eq:class_opt_disc_param}
		\begin{aligned}
			\underset{\nu^c_a \geq 0, a\in \Acal^c}{\max}\quad & { \sum\nolimits_{a\in \Acal^c} w_a^\star \nu^c_a} \\
			\mathrm{s.t.}\quad  \quad 
			&  \sum\nolimits_{a\in \Acal^c} (w^{\star}_a + \theta_a -\bar{w}^\star)  \leq  0\\
			& {  \sum\nolimits_{a\in \Acal^c} \nu^c_a  = m^c\Rd}
		\end{aligned}
	\end{equation}
	One can conclude that when $\theta_a = 0$ for all $a\in \Acal^c$, which corresponds to null discretization error, then the flows $f_a^{c\star}$ with $a\in \Acal^c$ are a solution to \eqref{eq:decomp_QP_param} by Lemma~\ref{lem:opt_prices} and it achieves an  objective of $m^c\Rd\bar{w}^\star$ at the solution. The objective of \eqref{eq:class_opt_disc_param} is continuous w.r.t.\ the decision variables and does not depend on the parameters $\theta_a$ with $a\in \Acal^c$, the feasible set is nonempty and compact for $\alpha\geq \alpha_1^\star$ (i.e., for $\theta_a$ with $a\in \Acal^c$ sufficiently close to zero), the feasible set is continuous (upper and lower hemicontinuous) at $\theta_a = 0$  for all $a\in \Acal^c$ since Condition~\ref{cond:distinct_payoff} holds. We are in the conditions of Berges's Maximum Theorem \cite[Theorem~17.31]{AliprantisBorder2006}, so the optimal objective is continuous w.r.t.\ the parameters $\theta_a$ with $a\in \Acal^c$ at $\theta_a = 0$  for all $a\in \Acal^c$. Finally, by increasing $\alpha$,  $\theta_a$  for all $a\in \Acal^c$ can be made arbitrarily close to zero and, by continuity, the solution to \eqref{eq:class_opt_disc_param} has an objective that can be made arbitrarily close to $m^c\Rd\bar{w}^\star$. It follows that for all $\epsilon>0$ that there exists $\alpha^\star$ such that for all $\alpha\geq \alpha^\star$, a solution to \eqref{eq:class_opt_disc_param} exists and its objective is at most $m^c\Rd\bar{w}^\star + m^c\Rd\epsilon$. One concludes that 
	\begin{equation}\label{eq:pol_payoff_ineq2}
		m^c\Rd(\bar{w}^\star +\epsilon)  \geq 	\sum\nolimits_{a\in \Acal^c} \nu^{c\star}_{a}w^{\star}_{a}.
	\end{equation}
	Combining \eqref{eq:pol_payoff_ineq1} and \eqref{eq:pol_payoff_ineq2} establishes the result.

%

\subsection{Proof of Theorem~\ref{th:opt_prices_MSNE}}\label{sec:proof_th_opt_prices_MSNE}

From Lemmas~\ref{lem:opt_prices} and~\ref{lem:opt_pol}, for all $\delta_1>0$ there exists a $\epsilon_1$-MSNE such that the action flows differ at most by $\delta_1$ w.r.t\  $f_a^{c\star}$ with $a\in \Acal^c$ in the game whose single-stage rewards are constant and equal to the payoffs at the optimum, i.e., $w^\star_a = w(a,\sigma^\star)$. From the first constraint of \eqref{eq:sys_opt}, the resource flows can also be made arbitrarily close to $\sigma^\star$, i.e., for all $\delta_2>0$ there exists a $\epsilon_2$-MSNE $\mu$ such that $|\sigma_r(\mu) -\sigma^\star_r| \leq \delta_2$ for all $r\in \Rcal$ for constant single-stage rewards. By the continuity of the resource reward functions under Assumption~\ref{ass:cont} and the fact that the resource flows can also be made arbitrarily close to $\sigma^\star$, for all $\delta_3$  there exists a $\epsilon_3$-MSNE $\mu$ such that $|\sigma_r(\mu) -\sigma^\star_r| \leq \delta_3$ for all $r\in [\Rsf]$ for the original game with resource reward functions $w_r$ for all $r\in [\Rsf]$. Finally, it follows from the (near) uniqueness of resource flows of $\epsilon$-MSNE in Theorem~\ref{th:unique_MSNE} that for all $\delta_4>0$ there is $\epsilon_4>0$ such that all $|\sigma_r(\mu) -\sigma^\star_r| \leq \delta_4$ for all $r\in [\Rsf]$ and all $\mu\in \MSNE_{\epsilon_4}(F,\phi)$. Take a generic $\epsilon>0$ and any $\mu \in \MSNE_\epsilon(F,\phi)$. Henceforth, $\Ocal(\epsilon)$ denotes a term of the form $l\epsilon$ for some $l>0$. For any $c\in [C]$, one can write $\bar w_c(\mu)$ as
\begin{equation*}
	\begin{split}
		& \sum_{k\in \Kcal}\sum_{u\in \UDcbar}\sum_{a\in \Acal^c} \mu^c[k,u]u(a|k)w(a,\sigma(\mu))\\
		 = &\!\!\!\!\!  \sum_{u\in \UDcbar: \mu^c[\Kcal,u]>\epsilon} \!\!\!\!\!\!\!\!\!\!\!\!  \mu^c[\Kcal,u] \!\sum_{k\in \Kcal}\sum_{a\in \Acal^c} \! \frac{\mu^c[k,u]}{\mu^c[\Kcal,u]}u(a|k)w(a,\sigma(\mu)) \!+\! \Ocal(\epsilon)\\
		 =& \!\!\!\!\! \sum_{u\in \UDcbar: \mu^c[\Kcal,u]>\epsilon}   \!\!\!\!\!\!\!\!\!\!\!\! \mu^c[\Kcal,u]\sum_{k\in \Kcal}\sum_{a\in \Acal^c} \eta^{c,u}(k)u(a|k)w(a,\sigma(\mu)) \!+\! \Ocal(\epsilon)\\
		 =& \!\!\!\!\! \sum_{u\in \UDcbar: \mu^c[\Kcal,u]>\epsilon}  \!\!\!\!\!\!\!\!\!\!\!\!  \mu^c[\Kcal,u]\sum_{k\in \Kcal}\sum_{a\in \Acal^c} \eta^{c,v^\star}(k)v^\star(a|k)w_a^\star \!+\! \Ocal(\epsilon)\\
		 = &\!\!\!\!\! \sum_{u\in \UDcbar: \mu^c[\Kcal,u]>\epsilon}  \!\!\!\!\!\!\!\!\!\! \!\!  \mu^c[\Kcal,u] \bar w_c^\star \!+\! \Ocal(\epsilon)\\
		 =& \;  \bar w_c^\star + \Ocal(\epsilon)
	\end{split}
\end{equation*}
where the equalities follow from: algebraic manipulation; \eqref{eq:epsMSNE_s} in Definition~\ref{def:MSNE}; the existence of a $\epsilon$-MSNE $\nu$ in Lemma~\ref{lem:opt_pol} that satisfies \eqref{eq:opt_pol_cond1} where $v^\star$ is such that $\nu^c[\Kcal,v^\star]>0$; the fact that $\nu$ satisfies \eqref{eq:opt_pol_cond2}; and algebraic manipulation, respectively.  \hfill$\qedsymbol$%


\section*{References}
\bibliographystyle{IEEEtran}
\bibliography{../../../../Papers/_bib/references-gt.bib,../../../../Papers/_bib/references-c.bib,../../../../Publications/bibliography/parsed-minimal/bibliography.bib}
	
\vspace{-2cm}
\begin{IEEEbiography}[{\includegraphics[width=1in,height=1.25in,clip,keepaspectratio]{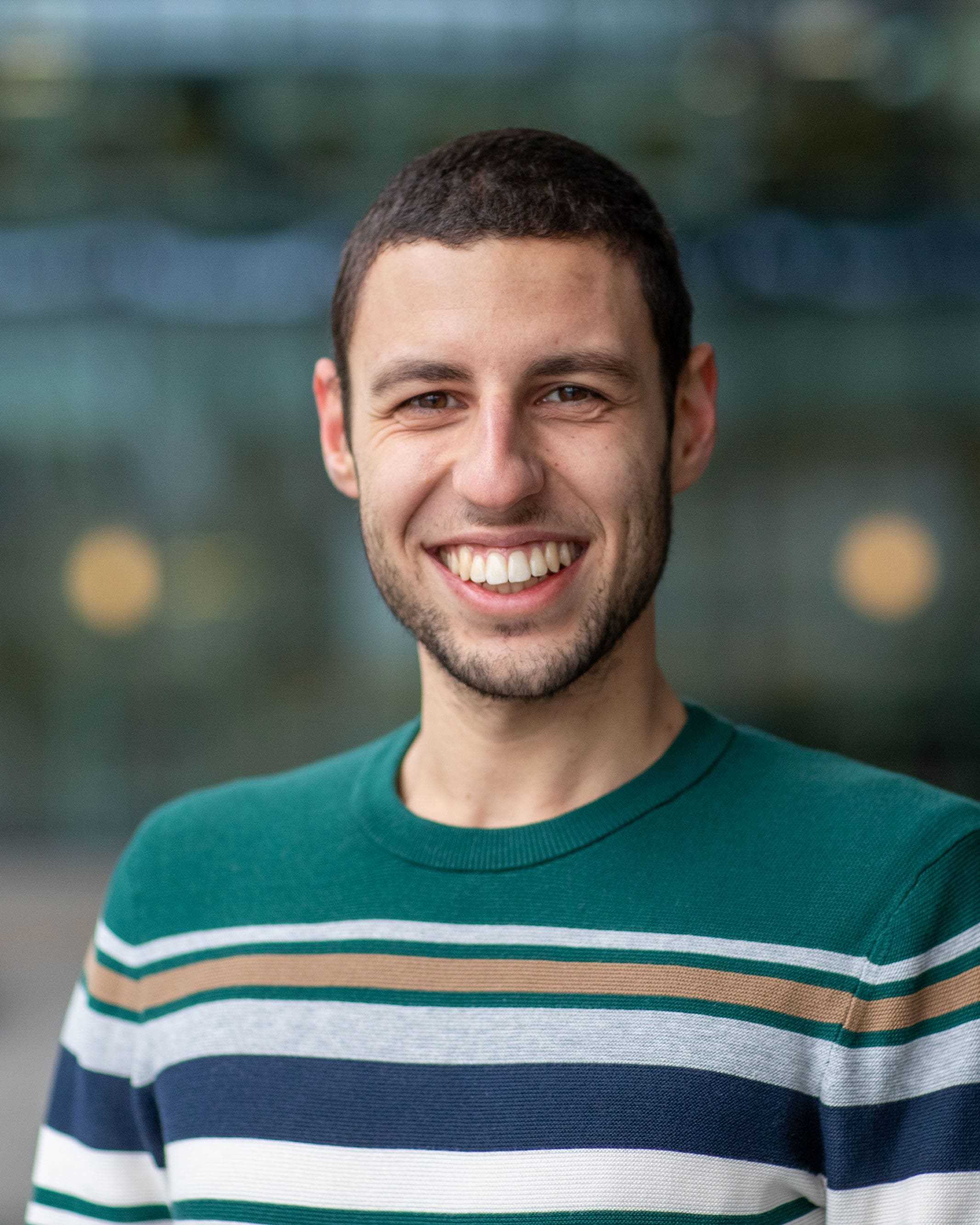}}]{Leonardo Pedroso} (Graduate Student Member, IEEE) received the M.Sc.\ degree in aerospace engineering from Instituto Superior T\'ecnico (IST), University of Lisbon (ULisboa), Portugal, in 2022. From 2019 to 2022, he held a research scholarship with the Institute for Systems and Robotics (ISR), IST, ULisboa. Since 2023, he is working toward the Ph.D.\ degree in mechanical engineering with the Control Systems Technology section, Eindhoven University of Technology, The Netherlands. Since 2024, he is working toward a second Ph.D.\ degree in aerospace engineering with the ISR, IST, ULisboa, Portugal.  He was the recipient of the 2024 Best M.Sc.\ Thesis Award by the Portuguese Automatic Control Association. His current research interests include mean field games and distributed control and estimation of ultra large-scale systems.\end{IEEEbiography}
\vspace{-1cm}
\begin{IEEEbiography}[{\includegraphics[width=1in,height=1.25in,clip,keepaspectratio]{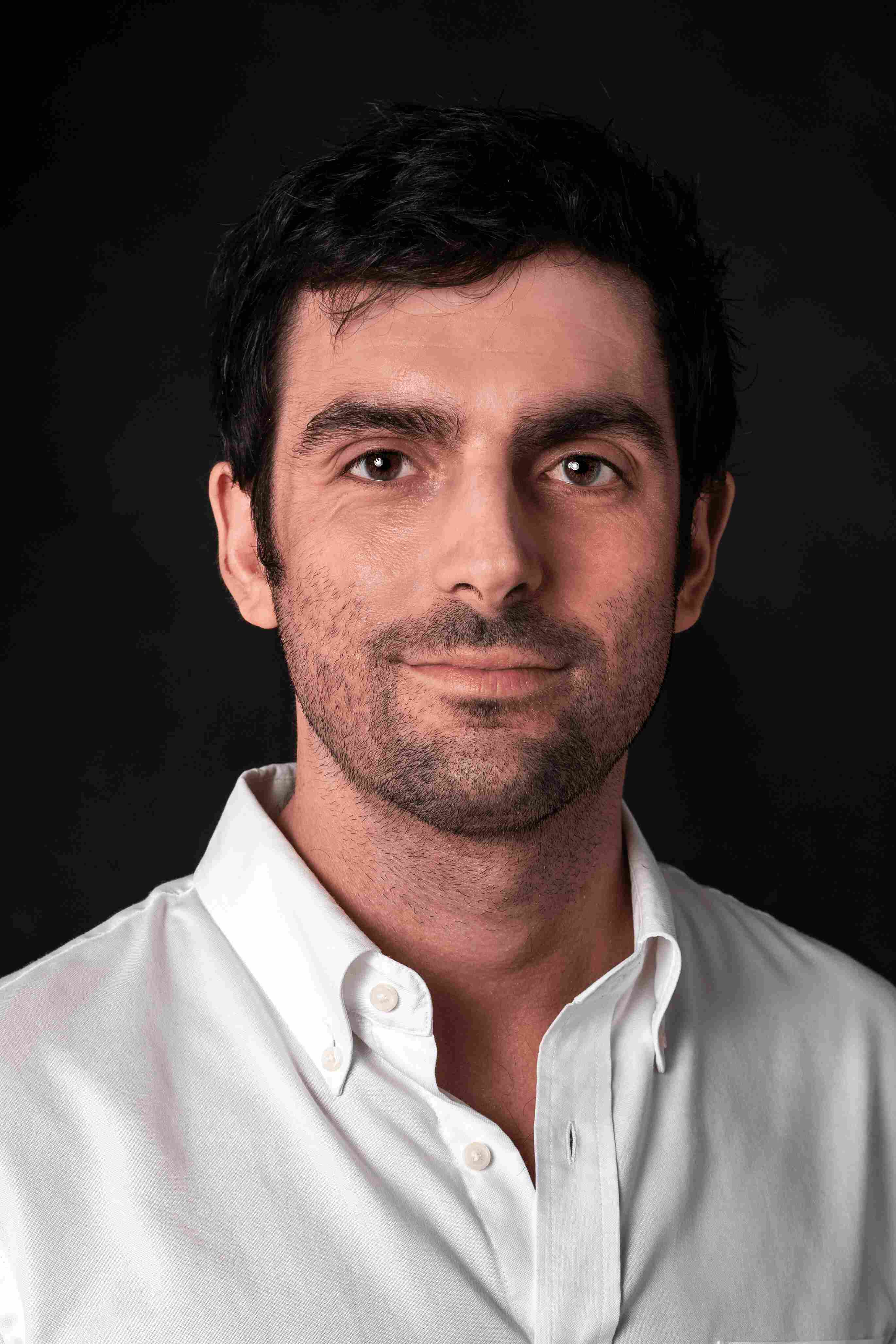}}]{Andrea Agazzi} received a B.Sc. in Physics from ETH Zurich, a M.Sc. in Theoretical and Mathematical Physics from Imperial College London (with distinction), and a Ph.D. in Theoretical Physics from the University of Geneva respectively in 2012, 2013 and 2017. From 2018 he was a Postdoctoral Associate and from 2019 Griffith Research Assistant Professor in the Mathematics Department at Duke University, USA. From 2022 to 2024 he was an Assistant Professor (tenure-track) at the University of Pisa, Italy. Since 2024, he has been an Associate Professor in the Department of Mathematics and Statistics at the University of Bern, Switzerland. His research focuses on stochastic analysis, interacting particle systems, and the mathematical foundations of deep learning.
\end{IEEEbiography}

\vspace{-1cm}
\begin{IEEEbiography}[{\includegraphics[width=1in,height=1.25in,clip,keepaspectratio]{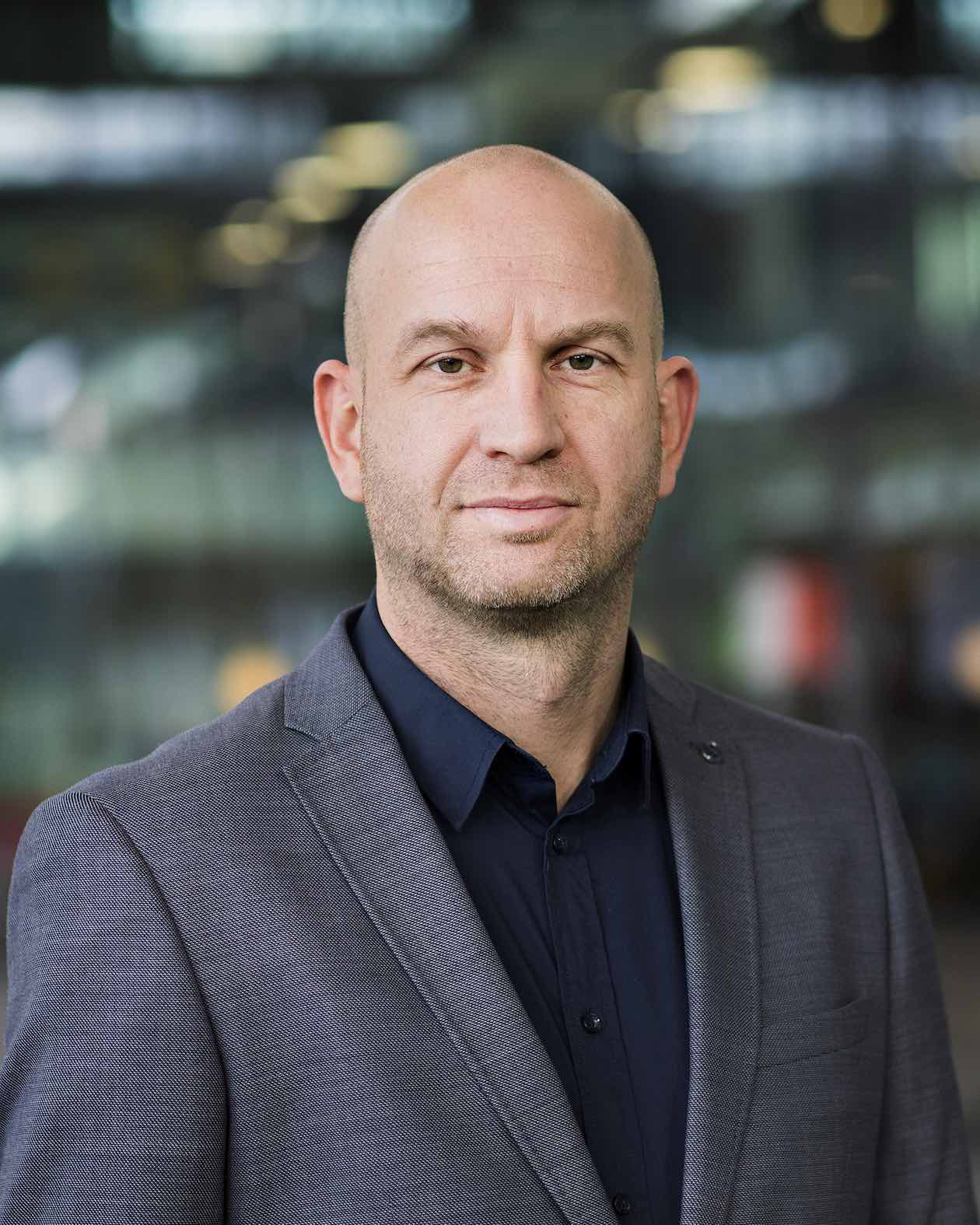}}]{W.P.M.H. Heemels} (Fellow, IEEE)  received M.Sc. (mathematics) and Ph.D. (EE, control theory) degrees (summa cum laude) from the Eindhoven University of Technology (TU/e) in 1995 and 1999, respectively. From 2000 to 2004, he was with the Electrical Engineering Department, TU/e, as an assistant professor, and from 2004 to 2006 with the Embedded Systems Institute (ESI) as a Research Fellow. Since 2006, he has been with the Department of Mechanical Engineering, TU/e, where he is currently a Full Professor and Vice-Dean. He held visiting professor positions at ETH, Switzerland (2001), UCSB, USA (2008) and University of Lorraine, France (2020). He is a Fellow of the IEEE and IFAC, and was the chair of the IFAC Technical Committee on Networked Systems (2017-2023). He served/s on the editorial boards of Automatica,  Nonlinear Analysis: Hybrid Systems (NAHS), Annual Reviews in Control, and IEEE Transactions on Automatic Control, and is the Editor-in-Chief of NAHS as of 2023. He was a recipient of a personal VICI grant awarded by NWO (Dutch Research Council) and recently obtained an ERC Advanced Grant. He was the recipient of the 2019 IEEE L-CSS Outstanding Paper Award and the Automatica Paper Prize 2020-2022. He was elected for the IEEE-CSS Board of Governors (2021-2023).   His current research includes hybrid and cyber-physical systems, networked and event-triggered control systems and model predictive control and their applications. \end{IEEEbiography}

\vspace{-1cm}
\begin{IEEEbiography}[{\includegraphics[width=1in,height=1.25in,clip,keepaspectratio]{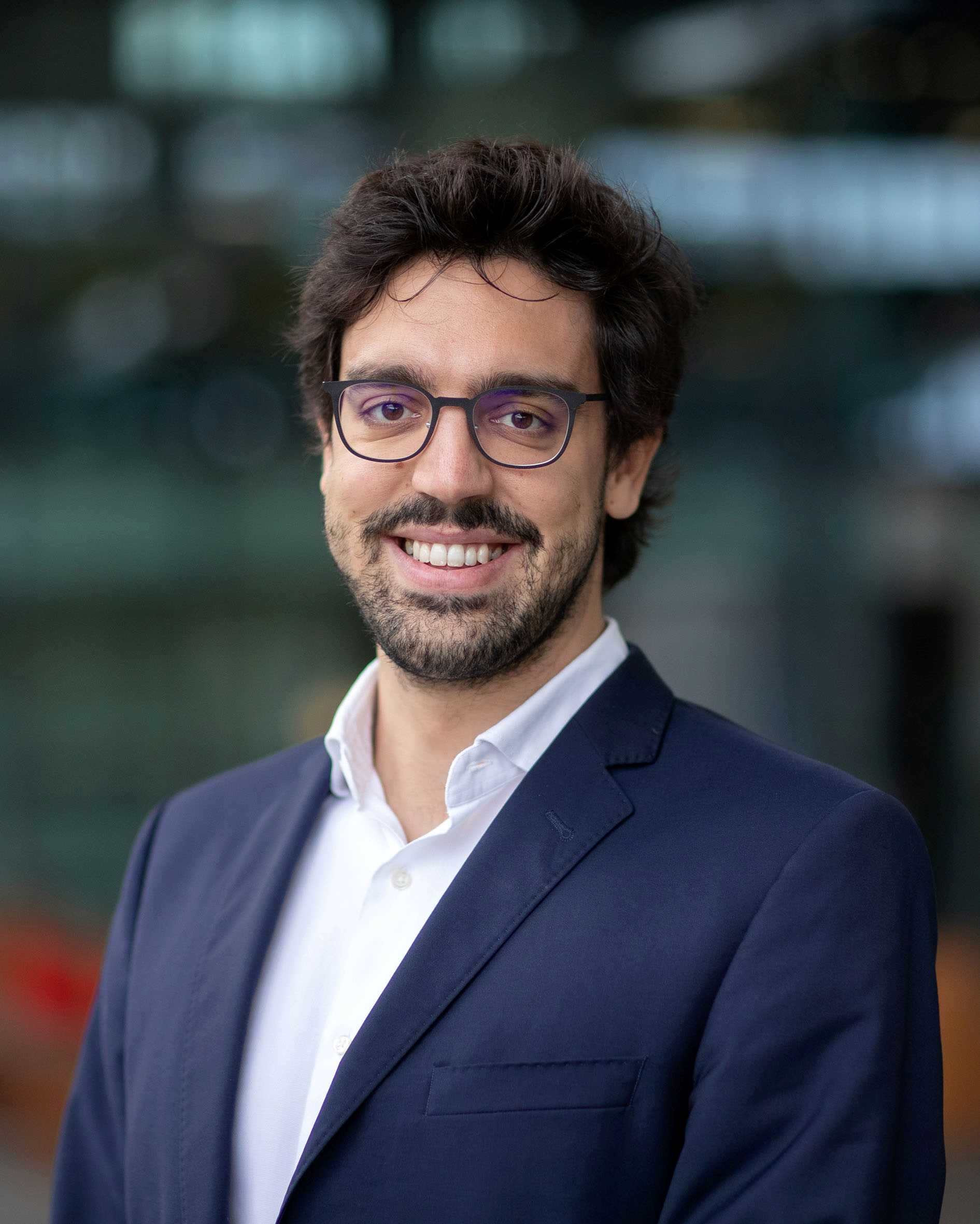}}]{Mauro Salazar} is an Associate Professor in the Control Systems Technology section at Eindhoven University of Technology (TU/e). He received the Ph.D. degree in Mechanical Engineering from ETH Z\"urich in 2019. Before joining TU/e he was a Postdoctoral Scholar in the Autonomous Systems Lab at Stanford University. Dr. Salazar’s research is at the interface of control theory and optimization, and is aimed at the development of a transdisciplinary set of tools for the design, the deployment and the operation of sustainable and societally responsible mobility systems. 
	Both his Master thesis and PhD thesis were recognized with the ETH Medal, and his papers were granted the Best Student Paper award at the 2018 Intelligent Transportation Systems Conference and the 2022 European Control Conference, and the Best Paper Award at the 2024 Vehicle Power and Propulsion Conference.
\end{IEEEbiography}


\end{document}